\documentclass[envcountsame,runningheads,orivec]{llncs}
\usepackage[T1]{fontenc}

\makeatletter
\RequirePackage[bookmarks,unicode,colorlinks=true]{hyperref}%
   \def\@citecolor{blue}%
   \def\@urlcolor{blue}%
   \def\@linkcolor{blue}%

\def\orcidID#1{\smash{\href{http://orcid.org/#1}{\protect\raisebox{-1.25pt}{\protect\includegraphics{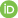}}}}}
\makeatother
\usepackage{bbding}
\usepackage{graphicx}
\usepackage{placeins}

\usepackage{cite}
\usepackage{listings}
\usepackage{graphicx,multirow,tabularx,makecell}
\usepackage{microtype}
\usepackage[compact]{titlesec}
\usepackage{paralist}
\usepackage[all]{xy}
\usepackage{stmaryrd}
\usepackage{ebproof,proof} 

\RequirePackage{symbolic-bisim}

\RequirePackage[textsize=tiny]{todonotes}
\RequirePackage{mathpartir}
\RequirePackage{cleveref}
\crefname{section}{Sec.}{Sec.(s)}
\crefname{appendix}{Appendix}{Appendices}
\crefname{figure}{Fig.}{Fig.(s)}
\crefname{theorem}{Thm.}{Thm.(s)}
\crefname{definition}{Def.}{Def.(s)}
\crefname{proposition}{Prop.}{Prop.(s)}
\crefname{corollary}{Cor.}{Cor.(s)}
\crefname{lemma}{Lem.}{Lem.(s)}
\crefname{example}{Ex.}{Ex.(s)}
\crefname{remark}{Rem.}{Rem.(s)}
\crefname{appendix}{App.}{App.(s)}
\usepackage{float}    

\usepackage{versions}
\excludeversion{exclproof}
\excludeversion{technical}
\newcommand\mytechnical[1]{}

\usepackage[misc]{ifsym}

\newcommand\cutout[1]{}

\begin{document}
\newcommand{\Sblock}[1]{\ma{\textup{\texttt{\{}} #1 \textup{\texttt{\}}}} }  
\newcommand{\funk}{\ma{\mathsf{function}}}                  
\newcommand{\Sfundef}[4]{\ma{\funk\,#1\textup{\texttt{(}}#2\textup{\texttt{)}}#3#4}}
\newcommand{\Sfundefto}[4]{\Sfundef{#1}{#2}{{\to}#3}{#4}}
\newcommand{\assignk}{\ma{\textup{\texttt{:=}}}}                     
\newcommand{\Svardecl}[2]{\ma{\letk\,#1\assignk#2}}
\newcommand{\Sassign}[2]{\ma{#1\assignk#2}}
\newcommand{\Scond}[2]{\ma{\ifk #1 #2}}
\newcommand{\switchk}{\ma{\mathsf{switch}}}                 
\newcommand{\casek}{\ma{\mathsf{case}}}                     
\newcommand{\defaultk}{\ma{\mathsf{default}}}               
\newcommand{\Sswitch}[2]{\ma{\switchk\,#1#2}}
\newcommand{\fork}{\ma{\mathsf{for}}}                       
\newcommand{\Sfor}[4]{\ma{\fork#1#2#3#4}}
\newcommand{\breakk}{\ma{\mathsf{break}}}                   
\newcommand{\continuek}{\ma{\mathsf{continue}}}             
\newcommand{\leavek}{\ma{\mathsf{leave}}}                   
\newcommand{\Efuncall}[2]{\ma{#1\textup{\texttt{(}}#2\textup{\texttt{)}}}}
\newcommand{\opcode}[1][]{\ma{\mathop{\ifempty{#1}{op}{#1}}}}
\newcommand{\objectk}{\ma{\mathsf{object}}}                 
\newcommand{\codek}{\ma{\mathsf{code}}}                     
\newcommand{\yulobj}[3]{\ma{\objectk\,#1\Sblock{\codek\,#2\,#3}}}
\newcommand{\datak}{\ma{\mathsf{data}}}                     
\newcommand{\yuldat}[2]{\ma{\datak\,#1\,#2}}

\newcommand{\bsconf}[4]{\ma{\langle {#4} \mid #1 \mathop{;} #2 \mathop{;} #3 \rangle}}
\newcommand{\bseval}{\ma{\Downarrow}}
\newcommand{\bsevalexp}{\ma{\Downarrow_{\mathsf{exp}}}}
\newcommand{\bsevalseq}{\ma{\Downarrow_{\mathsf{seq}}}}
\newcommand{\regk}{\ma{\mathsf{regular}}}
\newcommand{\modeM}{\ma{\mathbb{M}}}
\newcommand{\restrict}{\mathord{\upharpoonright}}
\newcommand{\nameN}{\ma{\mathcal{N}}}

\newcommand{\ssconf}[3]{\ma{\langle #1 \mid #2 \mathop{;} #3 \rangle}}
\newcommand{\ssframe}[1]{\ma{\llparenthesis #1 \rrparenthesis}}
\newcommand{\sscont}[1]{\ma{\llbracket #1 \rrbracket_{\mathsf{cnt}}} }
\newcommand{\ssbreak}[1]{\ma{\llbracket #1 \rrbracket_{\mathsf{brk}}} }
\newcommand{\opconf}[2]{\ma{\langle #1 \mid #2 \rangle}}

\newcommand{\toop}{\ma{\Downarrow_\mathsf{opc}}}

\newcommand{\objconf}[2]{\ma{\langle #1 \mid #2 \rangle}}

\newcommand{\datasizek}{\ma{\mathsf{datasize}}}
\newcommand{\dataoffsetk}{\ma{\mathsf{dataoffset}}}
\newcommand{\datacopyk}{\ma{\mathsf{datacopy}}}
\newcommand{\setdelta}{\ma{\mathsf{set\Delta}}}

\newcommand{\mem}{\ma{\mathsf{mem}}}
\newcommand{\store}{\ma{\mathsf{store}}}
\newcommand{\state}{\ma{\mathsf{state}}}

\newcommand{\Num}{\ma{\mathsf{num}}} 
\newcommand{\Str}{\ma{\mathsf{str}}} 
\newcommand{\Block}{\ma{\mathsf{Block}}} 
\newcommand{\Stmt}{\ma{\mathsf{Stmt}}} 
\newcommand{\Case}{\ma{\mathsf{Case}}} 
\newcommand{\Mode}{\ma{\mathsf{Mode}}} 
\newcommand{\ExpC}{\ma{\mathsf{ECxt}}} 
\newcommand{\StmtC}{\ma{\mathsf{SCxt}}} 
\newcommand{\Repo}{\ma{\mathsf{Repo}}} 

\newcommand{\ObjT}{\ma{\mathsf{obj}}} 
\newcommand{\AddrT}{\ma{\mathsf{addr}}} 
\newcommand{\SizeT}{\ma{\mathsf{size}}} 
\newcommand{\PosT}{\ma{\mathsf{pos}}} 
\newcommand{\DataFuns}{\ma{\mathsf{DataFun}}} 
\newcommand{\ObjSet}{\ma{\mathsf{Object}}} 
\newcommand{\DataSet}{\ma{\mathsf{Data}}} 

\newcommand{\ByteT}{\ma{\mathsf{byte}}} 
\newcommand{\KeccakT}{\ma{\mathsf{keccak}}} 

\newcommand{\varsof}[1]{\ma{\mathsf{vars}(#1)}} 
\newcommand{\funsof}[1]{\ma{\mathsf{funs}(#1)}} 

\newcommand{\yult}{\ma{\textsc{YulTracer}}} 
\raggedbottom

\title{An Operational Semantics for Yul}
\titlerunning{An Operational Semantics for Yul}
%
%
\author{Vasileios Koutavas\inst{1,3}%
\orcidID{0000-0002-3970-2486}
 \and
Yu-Yang Lin%
\inst{1}%
(\Envelope)%
\orcidID{0000-0001-5783-9454}
\and
Nikos Tzevelekos\inst{2}%
\orcidID{0000-0001-8509-8059}%
}

\authorrunning{V. Koutavas, Y.-Y. Lin, N. Tzevelekos}
%

\institute{Trinity College Dublin, Dublin, Ireland
\email{\{Vasileios.Koutavas,linhouy\}@tcd.ie} \and
Queen Mary University of London, London, UK
\email{nikos.tzevelekos@qmul.ac.uk} \and
Lero - Science Foundation Ireland Research Centre for Software, Limerick, Ireland
}
\maketitle              
\begin{abstract}
We present a big-step and small-step operational semantics for Yul\,---\,the intermediate language used by the Solidity compiler to produce EVM bytecode\,---\,in a mathematical notation that is congruous with the literature of programming languages, lends itself to language proofs, and can serve as a precise, widely accessible specification for the language. Our two semantics stay faithful to the original, informal specification of the language but also clarify under-specified cases such as void function calls.
Our presentation allows us to prove the equivalence between the two semantics.
We also implement the small-step semantics in an interpreter for Yul which avails of optimisations that are provably correct. We have tested the interpreter using tests from the Solidity compiler and our own.
  We envisage that this work will enable the development of verification and symbolic execution technology directly in Yul, contributing to the Ethereum security ecosystem, as well as aid the development of a provably sound future type system.

\keywords{Yul, Ethereum, Operational semantics, Programming languages, Formal methods}
\end{abstract}

  \section{Introduction}
  \label{sec:intro}
  Smart contracts are programs stored on a blockchain, executed within transactions, which read and write data on said blockchain.
They form an underpinning technology for decentralised finance applications by enabling the programming of autonomous services offering {lending, auction, new cryptocurrency creation, and }virtually any type of transactional facilities.
On Ethereum---the largest blockchain with the ability to execute smart contracts---smart contracts are typically written in a high-level programming language, predominantly Solidity, and then compiled down to Ethereum Virtual Machine (EVM) instructions.
Smart contracts already control significant amounts of cryptocurrency\footnote{As of 24 June 2024, DefiLlama (\url{https://defillama.com/chain/Ethereum}) lists a Total Value Locked (TVL) of USD 59.937b, equal to ETH 17.31m, on Ethereum.}, making them security-critical software.
Numerous documented exploits and bugs have resulted in the loss of billions of dollars worth of cryptocurrency for users~\cite{crypto-crimes-2024,eth-vulnerabilities,eth-bugs}, making smart contract correctness and security a problem of critical importance, {especially since these are immutable once deployed on the blockchain.}

The Solidity compiler team recently introduced Yul, an intermediate representation (IR) language sitting between Solidity and EVM bytecode, with one of its stated goals being
to help with the implementation of formal verification and optimisation techniques \cite{yul}.
Although such tools for Ethereum smart contracts exist, they either work with Solidity (e.g.~\textsc{solhint}\cite{tools:solidity:solhint}, \textsc{Slither}~\cite{tools:solidity:slither:paper,tools:solidity:slither},  \textsc{Echidna}~\cite{tools:solidity:echidna:paper,tools:solidity:echidna}, \textsc{Gambit}~\cite{tools:solidity:gambit}) or EVM bytecode (e.g. \textsc{Mythril}~\cite{tools:bytecode:mythril,tools:bytecode:mythril:survey}, \textsc{Oyente}~\cite{tools:bytecode:oyente,tools:bytecode:oyente:paper}, \textsc{Manticore}~\cite{tools:bytecode:manticore,tools:bytecode:manticore:paper}), making them hard to develop and maintain because of
the feature-rich, often evolving syntax of the former, or the low-level stack- and jump-based design of the latter.
Thus, Yul is designed on a small number of high-level constructs---loops, functions, conditional/switch and local variables---avoiding stack and jump complexity, and includes flexible, dialect-specific built-ins and data types.\footnote{Although currently Yul is used only with EVM primitives, it is designed to be parametric to a \emph{dialect} specifying the data types and machine instructions.}
All modern Solidity code can be compiled to Yul through the Solidity compiler, and almost all existing EVM bytecode can be decompiled to it~\cite{Smaragdakis-decompiler}.

In order for Yul to serve its intended purpose, it needs to have precise formal semantics on which formal verification techniques can be developed and optimisations can be proven correct.
The official documentation of Yul~\cite{yul}\footnote{Current latest version and the one accessed is v0.8.26 of the Solidity compiler} only provides an informal description of the language, its grammar, scoping rules, and its dynamic semantics in terms of an evaluation function written in pseudocode.
Although sufficient to understand the intent of the language constructs, the documentation is not a rigorous enough specification of Yul on which to base formal proofs.
As an example, the dynamic semantics in the documentation do not explicitly handle the case for calling functions with no return variables, which must be treated as statements rather than expressions (cf.~\cref{section:sem:expressions} and \cref{remark:void:functions}).

A number of mechanisations of the intended semantics of Yul in formal tools have appeared online. These include a shallow embedding of Yul into Dafny \cite{mechanisation:yul:dafny,mechanisation:yul:dafny:article}, and deep embeddings into the K-framework \cite{mechanisation:yul:K}, Isabelle/HOL \cite{mechanisation:yul:isabelle}, Lean \cite{mechanisation:yul:lean}, and ACL2 \cite{mechanisation:yul:acl2}.
Although these contribute significantly towards the goal of giving a solid mathematical footing for the semantics of Yul, they have inherent shortcomings.
{Besides not being peer-reviewed, their main drawback} is that in order to understand these mechanisations, one is required to have working knowledge {of} the formal tools used, something that, although increasingly more common, is certainly not universal among those who have received training in programming languages. Moreover, {the tools used} lend themselves more readily to mechanically formalising {the details of} various semantics rather than communicating {them} to experts.
{For example, the mechanisations of Yul contain explicit implementation of mathematical details that are otherwise conceptually straightforward (e.g. pattern-matching and manipulation of low-level dependencies such as sets, mappings, lists, etc), making them too lengthy and complex to serve as a high-level model of the language.}
Moreover, formalisation tools often incur engineering considerations that arise from their specific framework or underlying theories, which may require reshaping the presentation into a less standard form.
Finally, to even write a mechanisation of the Yul semantics one must first have a mental mathematical model of the language, which to our knowledge does not yet publicly exist. All this makes the existing Yul mechanisations hard to understand or trust as being equivalent to each other, which limits their suitability as definitive presentations of the semantics of Yul. For these reasons, we believe that the development of a diverse tool ecosystem for Yul would benefit from a more abstract, widely understood standard mathematical model thereof.

In this paper we contribute in this direction by providing a mathematical semantics of the language in textbook notation that, staying faithful to the intent of the original informal specification, can serve as a useful, widely understood, precise reference document. We hope our work to be useful for future programming language and formal verification research and implementations for Yul, the evolution of the language (e.g.~when the community decides to provide Yul with a type system and prove type soundness), and the validation of the existing mechanisations by comparing them with a more abstract model.

%

We present here an abstract syntax of Yul (\cref{sec:syntax}), together with both a big-step and small-step operational semantics (\cref{sec:sem}), and prove their correspondence {(\cref{sec:proof})}. We choose to present both styles of semantics because a big-step semantics is typically more intuitive and human-readable, and more readily relates to the informal specification of Yul, while a small-step semantics is easier to implement, especially in interpreters and symbolic execution tools.
Note that we present only the semantics for Yul, which is given separately from the semantics of EVM bytecode and dialect-specific notation such as objects. As mentioned earlier, this is because Yul is intended to be parametric to the specific dialect implemented. 
We also present a modular implementation of the semantics of Yul in the form of an interpreter that is parametric on a given dialect and avails of optimisations that are provably correct (\cref{sec:imp}). At the moment, this interpreter includes a subset of the semantics of the Shanghai upgrade of EVM as its only dialect, which is sufficient to evaluate the correctness of our semantics experimentally.
As such, our prototype interpreter only executes \emph{closed} programs that make use of a subset of the instructions available to the EVM. We evaluate our implementation with sample implementations of standard algorithms and tests from the Solidity compiler.
Finally, we discuss related work (\cref{sec:relwork}) and conclusions (\cref{sec:conclusion}).
Our work has been motivated by our ongoing efforts to implement a complete symbolic execution engine of EVM-dialect Yul, which would not have been possible without the semantics presented herein.


  \section{Yul Syntax}
  \label{sec:syntax}
  Yul is designed to be a highly readable language with simple syntax and semantics that is as direct to transform into bytecode as possible.
Because the EVM---the primary target of the language---is a stack machine, Yul achieves these goals by providing syntax that replaces control-flow and variable management operations that would usually be done via stack operations (e.g. $\mathsf{push}$, $\mathsf{pop}$) and jumps.

\begin{figure}[t]
  \[\begin{array}{r@{\;\;}r@{\,}c@{\,}l}
    \textsc{\Stmt:}  & S & \mis & \Sblock{S^*}
                           \mor \Sfundefto{x}{\vec x}{\vec x}{\Sblock{S^*}}
                           \mor \Svardecl{\vec x}{M}
                           \mor \Sassign{\vec x}{M}
                           \mor M
                           \mor \Scond{M}{\Sblock{S^*}}\\
                     & & & \mor \Sswitch{M}{(\casek\,v\Sblock{S^*})^*\defaultk\Sblock{S^*}}
                           \mor \Sfor{\Sblock{S^*}}{M}{\Sblock{S^*}}{\Sblock{S^*}}\\
                     & & & \mor \breakk
                           \mor \continuek
                           \mor \leavek
                           \\
    \textsc{\Exp:} & M & \mis & \Efuncall{x}{\vec M}
                           \mor \Efuncall{\opcode}{\vec M}
                           \mor x
                                \mor v
           \qquad \textsc{\Val:}\ u,v,c \qquad
                           \textsc{\Var:}\  x,y,z,f
    \end{array}\]\vspace{-5mm}
  \caption{Yul source-level syntax.}\label{fig:yul}
\end{figure}
We define the \emph{source-level} abstract syntax of Yul; that is, the syntax of Yul programs, in \cref{fig:yul}.
In this grammar we range over {(dialect-specific)} values {or constants} with $u,v,c$, variables with $x,y,z$, statements with $S$ and expressions with $M$, and variants thereof.
We use vector syntax (e.g.~$\vec{x}$) to mean a comma-separated, consecutively-indexed list of syntax ($x_1, x_2, \ldots, x_n$), and the Kleene star (e.g.~$S^*$) to mean consecutively-indexed juxtaposition of syntax without commas ($S_1 S_2 \ldots S_m$).
We next explain the syntactical constructs and their intention.
\begin{asparaitem}
\item \textbf{Block statement:} $\Sblock{S^*}$. These are sequences of statements that are to be executed one after the other unless a halting statement is encountered.
\item \textbf{Function definition:} \Sfundefto{x}{\vec x_{i}}{\vec x_{o}}{\Sblock{S^*}}.

  These contain the function name $x$, a vector of argument variables $\vec x_i$, a vector of return variables $\vec x_o$, and a function body block ${\Sblock{S^*}}$.
  Functions can return zero (hereafter referred to as void functions), one, or a tuple of values (hereafter referred to as multi-value functions).
  The use of curly braces in the function body is
  semantically important as blocks are leveraged to define the scope of variables in a function call.
  Actual arguments are assigned to $\vec x_i$ when entering---and return values are those assigned to $\vec x_o$ when exiting---the function block.
\item \textbf{Variable declaration/assignment:} \Svardecl{\vec x}{M} / \Sassign{\vec x}{M}. Their behaviour is identical except for the conditions involved: $\vec x$ cannot be in scope if declared {with a \textsf{let}-statement} and must be in scope for {standard} assignment.
  The syntax allows for multi-value assignment, with $M$ being an expression which must match the arity of $\vec x$.
  As tuples are not part of the source-level syntax, they can only occur as a result of calling a multi-value function. Tuples do not appear in \cref{fig:yul} but will be introduced in the runtime syntax in \cref{sec:sem}.
  The Yul specification \cite{yul} allows for uninitialised declarations.
  For simplicity we consider these to be syntactic sugar for function calls returning the correct arity of zeros.
\item \textbf{Conditional statement:} \Scond{M}{\Sblock{S^*}}. $M$ must evaluate to a constant that represents either $\true$ or $\false$, and $\Sblock{S^*}$ is a block for scoping reasons; {there are no else-branches in conditionals}.
\item \textbf{Switch statement:} \Sswitch{M}{(\casek\,v\Sblock{S^*})^*\defaultk\Sblock{S_d^*}}.
  If $M$ evaluates to constant $v_{i}$ then block $\Sblock{S^*_i}$ executes, otherwise $\Sblock{S_d^*}$ does.
\item \textbf{Loop:} \Sfor{\Sblock{S_i^*}}{M}{\Sblock{S_p^*}}{\Sblock{S_b^*}}.
  Loops comprise an initialisation block $\Sblock{S_i^*}$,
  a conditional expression $M$ which evaluates to \true or \false\ at each iteration, a post-iteration block $\Sblock{S_p^*}$ for updating loop variables, and a loop body $\Sblock{S_b^*}$.
\item \textbf{Halting statement:} \breakk, \continuek, \leavek. These are control-flow statements; \breakk and \continuek may only occur in loop bodies to exit the loop or skip the rest of the loop body and continue to the post-iteration block, respectively; \leavek only occurs inside function bodies to exit the current function (cf. \cref{def:halt:restrict}).
\item \textbf{Function call expression:} \Efuncall{x}{\vec M}.
  A function call is normally a part of an expression and returns a single value. However there are two special cases of note:
  (1) a function call can appear as the expression of a multi-value declaration/assignment and must return a tuple of values of the correct arity;
  (2) a void function call that appears at the place of a statement, in which case it must return $\regk$; this is the only allowed expression that can appear at the place of a statement.
  Arguments are evaluated \emph{right-to-left}.
\item \textbf{Opcode execution expression:} \Efuncall{\opcode}{\vec M}. These are calls to dialect-specific instructions external to Yul. Like function calls, arguments are evaluated right-to-left and return
  {zero, one, or multiple values.}
  Unlike function calls, opcodes may return
  exceptional values which form part of the dialect, not named by the Yul specification. We will treat them as abnormal program exit values in our small-step semantics in the following section.
\end{asparaitem}

\begin{example}
  To illustrate the source-level language, consider a naive recursive implementation of a Fibonacci term generator.
  Note that in Yul double forward slashes are comments.
\begin{lstlisting}[language=yul]
{ // Function that computes Fibonacci via naive recursion
  function fibonacci_rec(n) -> result {
    if lt(n,3) {
        result := 1
    }
    if gt(n,2) {
        result := add(fibonacci_rec(sub(n,1)),fibonacci_rec(sub(n,2)))
    }
  } // Call Fibonacci with 10
  let fib_10 := fibonacci_rec(10)
  mstore(0x00, fib_10)
}
\end{lstlisting}
  In the program above, function \textcode{fibonacci\_rec} computes the $n$-th Fibonacci number, where \textcode{lt}/\textcode{gt}, \textcode{add}/\textcode{sub} and \textcode{mstore} are comparison, arithmetic and memory instructions respectively. These are opcodes in the EVM dialect of Yul. Functions do not require a return statement. Instead, the defined output variable \textcode{result} is assigned and its value returned at the end. Nesting function calls within opcode/function calls is allowed, {and is only constrained by the stack limit of the underlying dialect (in EVM the stack limit is 1024 words)}. The program computes the 10th Fibonacci term and stores it at position \textcode{0x00} in memory.
\end{example}

\begin{remark}[Typing]
  Although simple checks are in place in the Solidity compiler, such as for the arity of function arguments/returns, Yul at the moment provides no formal specification for a type system. Its default dialect involves only EVM bytecode operations with a single value type of 256-bit unsigned integers ($\mathsf{U256}$). While the Yul source language involves literals of various kinds (such as hex literals), there is an implicit assumption that these are appropriately translated into values of the underlying dialect---e.g. $\mathsf{U256}$ for the EVM dialect.
\end{remark}

\begin{remark}[Dialects]\label{rem:dialects}
  As mentioned in the Introduction, because Yul is meant to be target-independent, it is designed to be modular with regards to multiple so-called dialects. Each dialect in Yul would in principle be defined in relation to the specific back-end bytecode that it is meant to compile into.
  For this purpose, dialects amount to a collection of machine opcodes
  that remain after excluding those associated with control flow and variable management, and the respective value types used by said opcodes.
  Thus, while we use $\true$ and $\false$
  values in our semantics, these do not necessarily represent boolean values; instead they are placeholders for whichever values are considered true and false in a given dialect.
  For instance, in the EVM dialect, zero represents $\false$ while non-zero values are considered $\true$.
  This is a conceptually simple mapping that entails a variety of subtle behaviours one has to be careful about in the EVM
  (e.g. negating a "true" value does not make it a "false" one, as negation is bit-wise and could produce another non-zero value).
\end{remark}

%
%
%
%
%
%
%
%
%
%
%
%


  \section{Yul Semantics}
  \label{sec:sem}
  In this section we present big- and small-step semantics for Yul. We shall present them together, grouped thematically under rules for blocks, variable manipulation, branching statements, loops, and expressions. For small-step semantics specifically, we additionally define top-level rules that carry around the evaluation contexts and handle external opcode instructions.

For both semantics, we shall use configurations of the form:
\[\bsconf{G}{L}{\nameN}{S}\]
with the following components:
\begin{asparaitem}
\item \textbf{Statement} $S$ (or \emph{term}) being evaluated. Statements are drawn from an extended \emph{operational syntax} defined in \cref{fig:yul2} (top), and discussed further below.
\item \textbf{Global environment} $G$ that is external to the semantics of Yul itself, {and in the context of smart contracts is meant to contain the state of the blockchain, the data of external calls to a smart contract, etc.;
  $G$ will be passed as state by the semantics of Yul and may change only by calls to dialect opcodes.}
\item \textbf{Local state} $L$ of all variables. This is a dictionary $L : x \mapsto v$. All variables in scope for the term being evaluated will be in $L$.
\item \textbf{Namespace} $\nameN$ of all functions visible to the term being evaluated. This is like a function repository $\nameN : x \mapsto (\vec x_i, \vec x_o, S_b)$ for all the functions in scope.
\end{asparaitem}

We mentioned that statements require an extended syntax;
  $\regk$  represents successful completion of a statement.
  Specific to the small-step semantics (cf.\ \Cref{def:small-step}) are
  special statements that
  arise from certain reductions.
  \begin{asparaitem}
\item \textbf{Block scoping:} $\Sblock{ \vec{S}}_{L}^{\nameN}$. {This is used to implement the scoping rules for blocks and are distinct from source-level blocks $\Sblock{S^*}$. Here $\nameN$ records functions and $L$ variables in scope before entering a block, needed when reinstating this scope after exiting the block.}
\item \textbf{Break/continue catching:} $\ssbreak{S}$ / $\sscont{S}$. These  catch $\breakk$ and $\continuek$ statements respectively within loop body blocks.
\item \textbf{Function call framing:} $\ssframe{S}_{L}^{\vec x}$. These encode function call frames by remembering the $L$ at call site and the return variables $\vec x$. At return, $L$ is used to restrict the domain and ensure correct scoping (similar to blocks) and $\vec x$ is used to retrieve the function return values.
  \end{asparaitem}

  We also extend expressions $M$ with tuples $\tuple{\vec v}$, needed for multi-value assignments and function returns, which are defined when $|\vec v|>1$.
  Finally, we introduce \emph{modes} $\modeM$,
  as in the Yul specification, which can signal regular termination ($\regk$), e.g.\ of a void function call, $\breakk$/$\continuek$ modes for for-loops, and $\leavek$ mode for returning early from functions.

  Modes also include \emph{external irregular modes} $\mathcal{M}$
 that may be returned by opcode evaluation. These are parametric to the Yul semantics and depend on the specific dialect implemented; e.g.\ in the EVM dialect, such modes may be generated by opcodes signalling exhaustion of gas.
We do not use external error modes in big-step semantics rules which define only successful executions.

%

\begin{figure}[t]
  \[\begin{array}{r@{\;\;}r@{\;\,}c@{\;\,}l}
    \textsc{\Stmt:}  & S & \mis & \dots
                           \mor \regk
                           \mor \mathcal{M}
                           \mor \Sblock{ \vec{S}}_{L}^{\nameN}
                           \mor \ssbreak{S}
                           \mor \sscont{S}
                           \mor \ssframe{S}_{L}^{\vec x}
                           \\
    \textsc{\Exp:} & M & \mis & \dots
                         \mor \tuple{\vec v}\\
                         \
          \textsc{\Mode:} & \modeM & \mis & \breakk
                           \mor \continuek
                           \mor \leavek
                           \mor \regk
                           \mor \mathcal{M}
                           \\
    \textsc{\ExMode:}  &\mathcal M
    \end{array}\]\vspace{-1mm}
  \hrule\vspace{-1mm}
    \[\begin{array}{r@{\;\;}r@{\,}c@{\,}l}
    \textsc{\StmtC:} & E & \mis & \hole
                           \mor \Sblock{ E,\vec S}_{L}^{\nameN}
                           \mor \Svardecl{\vec x}{E}
                           \mor \Sassign{\vec x}{E}
                           \mor E_{\Exp}
                           \mor \sscont{E}
                           \mor \ssbreak{E}
                           \mor \Scond{E_{\Exp}}{\Sblock{S^*}}
                           \\
                     & & & \mor \Sswitch{E_{\Exp}}{(\casek\,v\Sblock{S^*})^*\defaultk\Sblock{S^*}} \\
    \textsc{\ExpC:} & E_{\Exp} & \mis & \hole
                           \mor \Efuncall{x}{\vec M,E_{\Exp},\vec v}
                           \mor \Efuncall{\opcode}{\vec M,E_{\Exp},\vec v}
                           \mor \ssframe{E}_{L}^{\vec x}
  \end{array}\]\vspace{-5mm}
\caption{Yul operational syntax (omitting constructs shown in \cref{fig:yul}).}\label{fig:yul2}
  \end{figure}



In \cref{fig:yul2} (bottom) we define statement ($\textsc{\StmtC}$) and expression ($\textsc{\ExpC}$) evaluation contexts for our small-step semantics.
  Most evaluation contexts are derived in a standard way from the syntax of the language and ensure correct evaluation order in small-step semantics. For example block statements are evaluated left-to-right, whereas function and opcode call arguments are evaluated right-to-left.
  The inclusion of function call contexts $\ssframe{E}_{L}^{\vec x}$ encodes a call stack in our small-step semantics.
  As functions have a fresh variable scope when they are called, the surrounding scope $L$ of the caller must be recorded and reinstated at function return.
  The return variables must also be recorded to be returned to the caller.


  {In the following subsections we present the rules of the big- and small-step semantics of Yul, grouped by the kind of term they handle. These rules precisely define the predicates of the next two definitions.}

\begin{definition}[Big-step semantics]\label{def:big-step}
We define a big-step semantics on successful executions that do not result in external irregular modes. We write
\[
  \bsconf{G}{L}{\nameN}{S}
  \bseval
  \bsconf{G'}{L'}{\nameN'}{\modeM}
\]
  for the top-level big-step evaluation of $S$ that produces a {successful} execution.
We additionally define variants thereof for intermediate evaluations:
\begin{itemize}
  \item the evaluation of sequences {of statements comprising}
    blocks
\[\bsconf{G}{L}{\nameN}{\vec S}\bsevalseq \bsconf{G'}{L'}{\nameN'}{\modeM}\]
\item the evaluation of expressions, {which result in} $S' \in \{v, \tuple{\vec v} , \regk\}$
\[\bsconf{G}{L}{\nameN}{M}\bsevalexp \bsconf{G'}{L'}{\nameN'}{S'}\]
\end{itemize}
The above are concretely defined by the top rules in \cref{fig:sem:blocks} to \cref{fig:sem:exp}.
\end{definition}

\begin{remark}
  {Although it would have been possible to define big-step rules to deal with propagation of external irregular modes, we follow standard practice and} define $\bseval$ only on successful {computation},
  {excluding irregular modes. We do handle these modes in the small-step semantics, which more closely specifies an interpreter for the language.}
\end{remark}
\begin{definition}[Small-step semantics]\label{def:small-step}
We write
\[
  \ssconf{S}{G;L}{\nameN}
  \to
  \ssconf{S'}{G';L'}{\nameN'}
\]
for a small-step transition in our small-step semantics.
  When a transition does not change $G$ (as only opcode execution does) we omit it from the configurations.
The reduction rules
are shown in the bottom part of \cref{fig:sem:blocks} to \cref{fig:sem:exp}, with the additional top-level rules:
  \[\begin{array}{l@{\;\,}ll}
  \ssconf{E[S]}{G;L}{\nameN}
  \to
  \ssconf{E[S']}{G';L'}{\nameN'}
  & \text{if }
  \ssconf{S}{G;L}{\nameN}
  \to
  \ssconf{S'}{G';L'}{\nameN'}
  \\[1mm]
  \ssconf{E[\mathcal M]}{G;L}{\nameN}
  \to
  \ssconf{\mathcal M}{G';L'}{\nameN'}
  \end{array}\]
\end{definition}

The top-level rules decompose the term to an evaluation context with an inner statement in its hole that can either perform a Yul reduction (former rule), or is an irregular return mode that terminates the program (latter rule).
We next present the rules for small- and big-step semantics for each type of term.






\subsection{Block Statements}

\begin{figure}[t]
  \[\begin{array}{@{}cllll@{}}
    \irule[Block][block]{
      \nameN_1 = \funsof{\Sblock{S_1 .. S_n}}
      \\
      \bsconf{G}{L}{\nameN \uplus \nameN_1}{S_1 , .. , S_n}
      \bsevalseq
      \bsconf{G_1}{L_1}{\nameN \uplus \nameN_1}{\modeM}
    }{
      \bsconf{G}{L}{\nameN}{\Sblock{S_1 .. S_n}}
      \bseval
      \bsconf{G_1}{L_1\restrict L}{\nameN}{\modeM}
    }
    \\[2em]
    \irule[SeqEmpty][seqempty]{
    }{
      \bsconf{G}{L}{\nameN}{\varepsilon}
      \bsevalseq
      \bsconf{G}{L}{\nameN}{\regk}
    }
    \\[2em]
    \irule[SeqReg][seqreg]{
      \bsconf{G}{L}{\nameN}{S_1}
      \bseval
      \bsconf{G_1}{L_1}{\nameN}{\regk}
      \\
      \bsconf{G_1}{L_1}{\nameN}{\vec S}
      \bsevalseq
      \bsconf{G_2}{L_2}{\nameN}{\modeM}
    }{
      \bsconf{G}{L}{\nameN}{S_1 , \vec S}
      \bsevalseq
      \bsconf{G_2}{L_2}{\nameN}{\modeM}
    }
    \\[2em]
    \irule[SeqIrreg][seqirreg]{
      \bsconf{G}{L}{\nameN}{S_1}
      \bseval
      \bsconf{G_1}{L_1}{\nameN}{\modeM}
      \\
      \modeM\in\{\breakk, \continuek, \leavek\}
    }{
      \bsconf{G}{L}{\nameN}{S_1 , \vec S}
      \bsevalseq
      \bsconf{G_1}{L_1}{\nameN}{\modeM}
    }
    \\[2em]
    \irule[FunDef][fundef]{
    }{
      \bsconf{G}{L}{\nameN}{\Sfundefto{x}{\vec y}{\vec z}{\,S}}
      \bseval
      \bsconf{G}{L}{\nameN}{\regk}
    }\\[-2mm]
  \end{array}\]
  \hrule
  \[\begin{array}{l@{\;\;}l}
  \ssconf{\Sblock{S_1..S_n}}{L}{\nameN}
  \to
  \ssconf{\Sblock{S_1,..,S_n}_L^\nameN}{L}{\nameN \uplus \nameN_1}
    & \nameN_1 = \funsof{\Sblock{S_1..S_n}},~n>0
\\[0.5em]
  \ssconf{\Sblock{}}{L}{\nameN}
  \to
  \ssconf{\regk}{L}{\nameN}
  \\[0.5em]
  \ssconf{\Sblock{\regk,\vec S}_L^\nameN}{L_1}{\nameN_1}
  \to
  \ssconf{\Sblock{\vec S}_L^\nameN}{L_1}{\nameN_1}
  & \vec S \neq \epsilon
  \\[0.5em]
  \ssconf{\Sblock{\regk}_L^\nameN}{L_1}{\nameN_1}
  \to
  \ssconf{\regk}{L_1\restrict L}{\nameN}
  \\[0.5em]
  \ssconf{\Sblock{\modeM,\vec S}_L^\nameN}{L_1}{\nameN_1}
  \to
  \ssconf{\modeM}{L_1\restrict L}{\nameN}
  & \modeM \in \{\breakk,\continuek,\leavek\}

  \\[0.5em]
  \ssconf{\Sfundef{x}{\vec y}{\vec z}{\,S}}{L}{\nameN}
  \to
  \ssconf{\regk}{L}{\nameN}
  \end{array}\]\vspace{-5mm}
  \caption{Big-step (top) and small-step (bottom) block semantics.}\label{fig:sem:blocks}
\end{figure}

{
The first rules are those for blocks, shown in \cref{fig:sem:blocks}. Blocks sequentially execute statements and extend the namespace of variables and functions.
In the big-step semantics we handle function scoping in the \iref{block} rule, where the existing function namespace $\nameN$ is extended by the functions defined in the block (given by $\nameN_1=\funsof{\Sblock{S_1 .. S_n}}$), and block statements are evaluated under $\nameN\uplus\nameN_1$. Moreover, any new variables defined within the scope (discussed in the following subsection), are pruned in the resulting configuration of the rule by $L_1\restrict L$, which restricts the variable environment $L_1$ to only the variables in $L$.
Rules \iref{seqempty}, \iref{seqreg}, and \iref{seqirreg} handle the sequencing of statement execution, with the last one skipping all remaining statements in a sequence if a Yul mode other than $\regk$ is encountered, which causes early exit from a loop iteration, an entire loop, or a function body.
Function declarations are ignored by rule \iref{fundef}, as these are handled by the first rule.
}

On the small-step side, we make use of a context $\Sblock{\dots}_L^\nameN$ that remembers the variable and function namespaces to handle scoping.
{This context is introduced when a block is entered and removed (with the appropriate scope restriction) when exiting the block.
Note that, in this semantics and the Yul specification, functions do not need to be declared before being called, as long as they are in the same, or surrounding, block of the call.
}

{
\begin{remark}[Block optimisation]
  It is sound to simplify any
  $\Sblock{\ldots\Sblock{\Sblock{\cdot}_{L_1}^\nameN}_{L_2}^\nameN\ldots}_{L_n}^\nameN$
  to $\Sblock{\cdot}_L^\nameN$, provided $\dom{L_1} = \ldots = \dom{L_n}$ (cf.\ Corollary~\ref{cor:block:drop}), avoiding redundant block-exit reductions.
However, this may not provide a performance gain as calculating domain equality can be costlier than the additional reductions.
\end{remark}
}

\subsection{Variable Manipulation Statements}

\begin{figure}[t]
  \[\begin{array}{@{}cllll@{}}
    \irule[VarDecl][vardecl]{
      \bsconf{G}{L}{\nameN}{M}
      \bsevalexp
      \bsconf{G_1}{L_1}{\nameN}{v}
      \\
      x \not\in \dom{L}
    }{
      \bsconf{G}{L}{\nameN}{\Svardecl{x}{M}}
      \bseval
      \bsconf{G_1}{L_1[x \mapsto v]}{\nameN}{\regk}
    }
    \\[2em]
    \irule[TupleDecl][tvardecl]{
      \bsconf{G}{L}{\nameN}{M}
      \bsevalexp
      \bsconf{G_1}{L_1}{\nameN}{\tuple{\vec v}}
      \\
      \vec x \not\in \dom{L}
    }{
      \bsconf{G}{L}{\nameN}{\Svardecl{\vec x}{M}}
      \bseval
      \bsconf{G_1}{L_1[\vec x \mapsto \vec v]}{\nameN}{\regk}
    }
  \end{array}\]
  \hrule\vspace{2mm}
  \[\begin{array}{l@{\;\;}l}

  \ssconf{\Svardecl{\vec x}{\tuple{\vec v}}}{L}{\nameN}
  \to
  \ssconf{\regk} {L[\vec x \mapsto \vec v]}{\nameN}
  & \vec x \not\in\dom{L}


  \\[0.5em]
  \ssconf{\Svardecl{x}{v}}{L}{\nameN}
  \to
  \ssconf{\regk} {L[x \mapsto v]}{\nameN}
  & x \not\in\dom{L}

  \end{array}\]\vspace{-3mm}
  \caption{Big-step (top) and small-step (bottom) declaration semantics. Similar rules for assignments (e.g.\ $\Sassign{\vec x}{\tuple {\vec v}}$, with $\vec x\in\dom{L}$) omitted for economy.}\label{fig:sem:vars}
\end{figure}

{
Single- and multi-value variable declaration are handled by the rules in \cref{fig:sem:vars}.
Assignment rules are similar, with the only differences being that the statement is missing the $\letk$-keyword, and the side condition in all rules requires $x \in\dom{L}$ (or $\vec x \in\dom{L}$) instead of $x\not\in\dom{L}$ (resp., $\vec x\not\in\dom{L}$).
These rules behave as expected: they add or update $L$ bindings.
Unlike functions, variables must be declared before use.
Only function and opcode calls may return multiple values.
}

\subsection{Branching Statements}

\begin{figure}[t]
  \[\begin{array}{@{}cllll@{}}
    \irule[IfF][iffalse]{
      \bsconf{G}{L}{\nameN}{M}
      \bsevalexp
      \bsconf{G_0}{L_0}{\nameN}{\false}
    }{
      \bsconf{G}{L}{\nameN}{\Scond{M}{{ S}}}
      \bseval
      \bsconf{G_0}{L_0}{\nameN}{\regk}
    }
    \\[2em]
    \irule[IfT][iftrue]{
      \bsconf{G}{L}{\nameN}{M}
      \bsevalexp
      \bsconf{G_0}{L_0}{\nameN}{\true}
      \\
      \bsconf{G_0}{L_0}{\nameN}{{ S}}
      \bseval
      \bsconf{G_1}{L_1}{\nameN}{\modeM}
    }{
      \bsconf{G}{L}{\nameN}{\Scond{M}{{ S}}}
      \bseval
      \bsconf{G_1}{L_1}{\nameN}{\modeM}
    }
    \\[2em]
    \irule[SwD][switchdef]{
      \bsconf{G}{L}{\nameN}{M}
      \bsevalexp
      \bsconf{G_0}{L_0}{\nameN}{v}
      \\
      \bsconf{G_0}{L_0}{\nameN}{S_d}
      \bseval
      \bsconf{G_1}{L_1}{\nameN}{\modeM}
    }{
      \bsconf{G}{L}{\nameN}{\Sswitch{M}{C_1 .. C_n \defaultk S_d}}
      \bseval
      \bsconf{G_1}{L_1}{\nameN}{\modeM}
    }
    \\[2em]
    \irule[SwC][switchcase]{
      \bsconf{G}{L}{\nameN}{M}
      \bsevalexp
      \bsconf{G_0}{L_0}{\nameN}{c_{n+1}}
      \\
      \bsconf{G_0}{L_0}{\nameN}{S_{n+1}}
      \bseval
      \bsconf{G_1}{L_1}{\nameN}{\modeM}
    }{
      \bsconf{G}{L}{\nameN}{\Sswitch{M}{C_1 .. C_{n+1}.. \defaultk S_d}}
      \bseval
      \bsconf{G_1}{L_1}{\nameN}{\modeM}
    }\\[-2mm]
  \end{array}\]
  \hrule
  \[\begin{array}{l@{\;\,}ll}
\\[-2mm]
  \ssconf{\Scond{\false}{\,S}}{L}{\nameN}
  \to
  \ssconf{\regk}{L}{\nameN}
  \qquad\qquad
  \ssconf{\Scond{\true}{\,S}}{L}{\nameN}
  \to
  \ssconf{S}{L}{\nameN}

  \\[0.5em]
  \ssconf{\Sswitch{v\,}{C_1 .. C_n \defaultk\, S_d}}{L}{\nameN}
  \to
  \ssconf{S_d}{L}{\nameN}
  \\[0.5em]
  \ssconf{\Sswitch{c_{n+1}\,}{C_1 .. C_{n+1}..  \defaultk\, S_d}}{L}{\nameN}
  \to
  \ssconf{S_{n+1}}{L}{\nameN}

  \end{array}\]\vspace{-3mm}
\caption{Big-step (top) and small-step (bottom) branching semantics. For switch rules we assume side conditions:
      $C_i = \casek\,c_i S_i$ and $v,c_{n+1} \not\in \{c_1 ,..,c_n\}$.
}\label{fig:sem:cond}
\end{figure}

Control branching in Yul is done by conditional and switch statements.
Conditionals make use of $\true$ and $\false$, which are not necessarily boolean values (as Yul is parametric in the dialect types) but, instead,  placeholders for values that are to be interpreted correspondingly. We assume a dialect comes equipped with predicates that let us identify which values are $\true$ and $\false$.


{
The rules for these statements are given in \cref{fig:sem:cond}.
On the big-step side, we have standard \iref{iffalse} and \iref{iftrue} rules for conditionals, and \iref{switchdef} and \iref{switchcase} for case-switching. All switch statements have a default case that is called if {the deciding expression $M$ evaluates to} $v$ {which} does not match any $c_i$ in the listed cases. The same behaviour is captured by standard small-step rules.
Note that source-level syntax requires inner statements in conditional and switch statements be (potentially empty) blocks, with their own scope.
}

\subsection{Loop Statements}

\begin{figure}[t]
  \[\begin{array}{@{}cllll@{}}
    \irule[ForInit][forinit]{
      \bsconf{G_1}{L_1}{\nameN}{\Sblock{S_1 .. S_n\,\Sfor{\Sblock{}}{M}{S_{p}}{ S_{b}}}}
      \bseval
      \bsconf{G_2}{L_2}{\nameN}{\modeM}
      \\
      n > 0
    }{
      \bsconf{G_1}{L_1}{\nameN}{\Sfor{\Sblock{S_1 .. S_n}}{M}{{S_{p}}}{{ S_{b}}}}
      \bseval
      \bsconf{G_2}{L_2}{\nameN}{\modeM}
    }
    \\[2em]
    \irule[ForFalse][forfalse]{
      \bsconf{G}{L}{\nameN}{M}
      \bsevalexp
      \bsconf{G_1}{L_1}{\nameN}{\false}
    }{
      \bsconf{G}{L}{\nameN}{\Sfor{\Sblock{}}{M}{{ S_{p}}}{{ S_{b}}}}
      \bseval
      \bsconf{G_1}{L_1}{\nameN}{\regk}
    }
    \\[2em]
    \irule[ForHalt1][forhalt1]{
      \bsconf{G}{L}{\nameN}{M}
      \bsevalexp
      \bsconf{G_1}{L_1}{\nameN}{\true}
      \\
      \bsconf{G_1}{L_1}{\nameN}{{ S_{b}}}
      \bseval
      \bsconf{G_2}{L_2}{\nameN}{\modeM_{lb}}
    }{
      \bsconf{G}{L}{\nameN}{\Sfor{\Sblock{}}{M}{{ S_{p}}}{{ S_{b}}}}
      \bseval
      \bsconf{G_2}{L_2}{\nameN}{\modeM_{lb}'}
    }
    \\[2em]
    \irule[ForHalt2][forhalt2]{
      \bsconf{G}{L}{\nameN}{M}
      \bsevalexp
      \bsconf{G_1}{L_1}{\nameN}{\true}
      \\
      \bsconf{G_1}{L_1}{\nameN}{{ S_{b}}}
      \bseval
      \bsconf{G_2}{L_2}{\nameN}{\modeM_{\neq lb}}
      \\
      \bsconf{G_2}{L_2}{\nameN}{{ S_{p}}}
      \bseval
      \bsconf{G_3}{L_3}{\nameN}{\leavek}
    }{
      \bsconf{G}{L}{\nameN}{\Sfor{\Sblock{}}{M}{{ S_{p}}}{{ S_{b}}}}
      \bseval
      \bsconf{G_3}{L_3}{\nameN}{\leavek}
    }
    \\[2em]
    \irule[ForLoop][forloop]{
      \bsconf{G}{L}{\nameN}{M}
      \bsevalexp
      \bsconf{G_1}{L_1}{\nameN}{\true}
      \\
      \bsconf{G_1}{L_1}{\nameN}{{ S_{b}}}
      \bseval
      \bsconf{G_2}{L_2}{\nameN}{\modeM_{\neq lb}}
      \\
      \bsconf{G_2}{L_2}{\nameN}{{ S_{p}}}
      \bseval
      \bsconf{G_3}{L_3}{\nameN}{\regk}
      \quad
      \bsconf{G_3}{L_3}{\nameN}{\Sfor{\Sblock{}}{M}{{ S_{p}}}{{ S_{b}}}}
      \bseval
      \bsconf{G_4}{L_4}{\nameN}{\modeM}
    }{
      \bsconf{G}{L}{\nameN}{\Sfor{\Sblock{}}{M}{{ S_{p}}}{{ S_{b}}}}
      \bseval
      \bsconf{G_4}{L_4}{\nameN}{\modeM}
    }
  \end{array}\]
  \hrule
  \[\begin{array}{l@{\;\;}l}
      \ssconf{\Sfor{\Sblock{S_1 .. S_n}}{M}{{S_{p}}}{{ S_{b}}}}{L}{\nameN}
  \to
  \ssconf{
        \Sblock{S_1 .. S_n\,\Sfor{\Sblock{}}{M}{S_{p}}{ S_{b}}}
  }{L}{\nameN}
& n>0
  \\[0.5em]
  \ssconf{\Sfor{\Sblock{}}{M}{{S_{p}}}{{ S_{b}}}}{L}{\nameN}
  \to
  \ssconf{
    \ssbreak{
    \Scond{M}{
    \Sblock{
      \sscont{S_b} S_p
      \Sfor{\Sblock{}}{M}{{S_{p}}}{{ S_{b}}}}
    }
    }
  }{L}{\nameN}

  \\[0.5em]
  \ssconf{\sscont{\modeM}}{L}{\nameN}
  \to
  \ssconf{\modeM}{L}{\nameN}
  & \modeM \in \{\regk,\breakk,\leavek\}
  \\[0.5em]
  \ssconf{\sscont{\continuek}}{L}{\nameN}
  \to
  \ssconf{\regk}{L}{\nameN}

  \\[0.5em]
  \ssconf{\ssbreak{\modeM}}{L}{\nameN}
  \to
  \ssconf{\modeM}{L}{\nameN}
  & \modeM \in \{\regk,\leavek\}
  \\[0.5em]
  \ssconf{\ssbreak{\breakk}}{L}{\nameN}
  \to
  \ssconf{\regk}{L}{\nameN}

  \end{array}\]\vspace{-5mm}
\caption{Big-step (top) and small-step (bottom) loop semantics. Side conditions: $(\modeM_{lb},\modeM_{lb}')\in\{(\leavek,\leavek),(\breakk,\regk)\}$, whereas
$\modeM_{\neq lb}\in\{\regk,\continuek\}$.}\label{fig:sem:loop}
\end{figure}

Loops in Yul consist of four parts: an initialisation block, a condition, a post-iteration block, and a loop-body block. In well-formed Yul programs \breakk and \continuek may appear only in a loop body. On the other hand, \leavek may appear in any part of a loop, so long as the loop itself occurs inside a function.
The semantics for loops is given in \cref{fig:sem:loop}.

Here we once again make use of the semantics of blocks, this time to handle the sequencing of loops. \iref{forinit} refactors the initialisation block out of the loop and places the whole loop in a new block that prefixes the initialisation statements\,---\,the small-step rule mirrors this exactly. \iref{forhalt2} handles break and leave statements within the loop body, while \iref{forhalt2} handles leave statements within the post-iteration. Note that \continuek is handled implicitly by the big-step rules; we piggyback off \iref{seqirreg} in \cref{fig:sem:blocks} to exit out of the loop body if \continuek -- i.e. an irregular mode\,---\,is encountered. The small-step side handles irregular modes more explicitly:
special evaluation contexts $\ssbreak{\dots}$ and $\sscont{\dots}$ are used to encode the frames that check for \breakk and \continuek respectively. Lastly, note that while \iref{forloop} directly executes the loop body and otherwise breaks with \iref{forfalse}, the small-step equivalent further leverages the semantics of conditional statements to encode this behaviour.

{
\begin{remark}[Break optimisation]
  Multiple small-step unrollings of a for-loop result in statements containing the pattern $\ssbreak{\Sblock{\ssbreak{\Sblock{\ssbreak{\ldots\Sblock{\ssbreak{\cdot}}_{L_n}^\nameN\ldots}}_{L_2}^\nameN}}_{L_1}^\nameN}$
  with $\dom{L_1}=\ldots=\dom{L_n}$.
Using our semantics we prove that such statements can be simplified by dropping the inner break and block contexts
  (cf.~\cref{cor:break:drop,lem:redundant}). This has brought significant speedup in the implementation of our interpreter.
  That is, for any $E$ that contains only interleavings of $\ssbreak{\cdot}$ and empty $\Sblock{\cdot}_{L_j}^{\nameN}$ contexts with the same $\nameN$ and domain in all $L_j$'s:
\begin{align*}
  &\ssconf{\ssbreak{E[\ssbreak{S'}]}}{L'}{\nameN'}
  \to^*
  \ssconf{\modeM}{L''}{\nameN''}
  \text{ iff }
  \ssconf{\ssbreak{E[{S'}]}}{L'}{\nameN'}
  \to^*
  \ssconf{\modeM}{L''}{\nameN''}
  \\
  &\ssconf{\Sblock{E[\Sblock{S'}_{L_2}^\nameN]}_{L_1}^\nameN}{L'}{\nameN'}
  \to^*
  \ssconf{\modeM}{L''}{\nameN''}
  \text{ iff }
  \ssconf{\Sblock{E[{S'}]}_{L_1}^\nameN}{L'}{\nameN'}
  \to^*
  \ssconf{\modeM}{L''}{\nameN''}
\end{align*}
\end{remark}
}

\subsection{Expressions}\label{section:sem:expressions}

Finally, we define the semantics of expressions, handling variable look-up, function calls and opcode calls. Because of this, they  either return values, \regk, or an external irregular mode $\mathcal{M}$ if handled. The rules are given in \cref{fig:sem:exp}.

\begin{figure}[t]
  \[\begin{array}{@{}cllll@{}}
    \irule[Ident][ident]{
    }{
      \bsconf{G}{L}{\nameN}{x}
      \bsevalexp
      \bsconf{G}{L}{\nameN}{L(x)}
    }
    \\[2em]
    \irule[FunCall][funcall]{
      \bsconf{G}{L}{\nameN}{M_n}
      \bsevalexp
      \bsconf{G_1}{L_1}{\nameN}{v_n}\ \dots\
      \bsconf{G_{n-1}}{L_{n-1}}{\nameN}{M_1}
      \bsevalexp
      \bsconf{G_n}{L_n}{\nameN}{v_1}
      \\
      \bsconf{G_n}{L_f}{\nameN}{{S_b}}
      \bseval
      \bsconf{G''}{L'}{\nameN}{\modeM}
    }{
      \bsconf{G}{L}{\nameN}{\Efuncall{f}{M_1 , .. , M_n}}
      \bsevalexp
      \bsconf{G''}{L_n}{\nameN}{\tuple{L'(z_1),..,L'(z_m)}}
    }
    \\[2em]
    \irule[OpcCall][opcall]{
      \bsconf{G}{L}{\nameN}{M_n}
      \bsevalexp
      \bsconf{G_1}{L_1}{\nameN}{v_n}
      \ \dots\
      \bsconf{G_{n-1}}{L_{n-1}}{\nameN}{M_1}
      \bsevalexp
      \bsconf{G_n}{L_n}{\nameN}{v_1}
      \\
      \opconf{\Efuncall{\opcode}{v_1 , .. , v_n}}{G_n}
      \toop
      \opconf{v'_1 , .. , v'_m}{G''}
    }{
      \bsconf{G}{L}{\nameN}{\Efuncall{\opcode}{M_1, .. , M_n}}
      \bsevalexp
      \bsconf{G''}{L_n}{\nameN}{\tuple{v_1',..,v_m'}}
    }
  \end{array}\]
  \hrule
  \[\begin{array}{l@{\;\;}ll}

  \ssconf{x}{L}{\nameN}
  \to
  \ssconf{L(x)}{L}{\nameN}
  & x \in L

  \\[0.5em]
  \ssconf{\Efuncall{f}{\vec v}}{L}{\nameN}
  \to
  \ssconf{\ssframe{S_b}^{\vec z}_{L}}{L_f}{\nameN}



  \\[0.5em]
  \ssconf{{\ssframe{\modeM}^{z_1 , .. , z_m}_L}}{L'}{\nameN}
  \to
  \ssconf{\tuple{L'(z_1), .. ,L'(z_m)}}{L}{\nameN}
  & \modeM \in \{\leavek,\regk\}




  \\[0.5em]
    {
  \opconf{\Efuncall{\opcode}{\vec v}}{G;L;\nameN}
  \to
  \opconf{S'}{G';L;\nameN}
  \qquad
    }  &
    \nbox{
    \text{if }
    \opconf{\Efuncall{\opcode}{\vec v}}{G}
    \toop
    \opconf{S'}{G'}
    \\\text{and }
    S'\in\{\tuple{\vec v}, \mathcal M\} }
%
  \end{array}\]\vspace{-6mm}
\caption{Big-step (top) and small-step (bottom) expression semantics. Side-conditions:
    {$\nameN(f) = ((y_1,..,y_n),(z_1,..,z_m),S_{b})$,
    $L_f = \{(y_1,..,y_n) \mapsto (v_1, .. , v_n)\} \uplus \{\{z_1,..,z_m\} \mapsto 0\}$.}
 For uniformity, we slightly abuse tuple notation to write $\tuple{v}=v$ and $\tuple{}=\regk$.}\label{fig:sem:exp}
\end{figure}
\iref{ident} and its small-step counterpart behave as expected; they look up the value of $x$ within $L$. Function calls on the other hand are less standard. Instead of substitution for function application, calls assign in $L$ the correct value for every argument and then leverage the semantics of blocks to set up the namespace. Similarly, for returns, instead of a return statement, functions assign to the return variables and look up their values in the local $L$, again leveraging block semantics to clean up the namespace. Note that in \iref{funcall}, functions may return either a sequence of values if there are return variables, or \regk if the function is not expected to return anything (analogous to $\mathsf{void}$ functions in C-like languages). \iref{opcall} mirrors function calls. Note that we do not handle external modes in the big-step semantics. For the small-step transitions involving function calls we introduce a special context $\ssframe{\dots}^{\vec z}_{L}$. Much like with other special contexts introduced in the small-step semantics, its purpose is to encode within the evaluation context all the frames that form the call stack, including the local state $L$ to reinstate and the return variables $\vec z$.

{
\begin{remark}[Void return and modes]\label{remark:void:functions}
The evaluation function provided in the Yul specification~\cite{yul} is undefined for void function calls. This introduces ambiguity in the
intended details of the Yul semantics.
  We decided to resolve this in a
  consistent manner to~\cite{yul}, making modes a subset of statements, produced by evaluation, and for \regk to be returned by void functions.
  Since such functions are called in block statements, it was necessary that they return \regk in order for the evaluation of subsequent block statements to occur.
\end{remark}
}

  \section{Semantic Equivalence}
  \label{sec:proof}
  Here we show that the big-step and small-step semantics are equivalent.
A necessary result is a congruence lemma for small-step semantics.

\begin{lemma}[Congruence]\label{lem:congruence}
Given any evaluation context $E$ and a reduction sequence
$\bsconf{G}{L}{\nameN}{S} \to^* \bsconf{G'}{L'}{\nameN'}{S'}$
then it is the case that:
\begin{align*}
&\bsconf{G}{L}{\nameN}{S} \to^* \bsconf{G'}{L'}{\nameN'}{S'}
\implies
\bsconf{G}{L}{\nameN}{E[S]} \to^* \bsconf{G'}{L'}{\nameN'}{E[S']}
\end{align*}
\end{lemma}
\begin{proof}
This follows by induction on the length of the reduction.
\qed
\end{proof}


\begin{theorem}[Semantics Equivalence]\label{theorem:equiv}
Given statement $S$ in source syntax, global environment $G$, local state $L$, and function namespace $\nameN$,
  then, for a configuration $\bsconf{G}{L}{\nameN}{S}$, statements (1) and (2) below are equivalent.
\begin{enumerate}[(1)]
\item $\bsconf{G}{L}{\nameN}{S} \bseval \bsconf{G'}{L'}{\nameN'}{\regk}$
\item $\bsconf{G}{L}{\nameN}{S} \to^* \bsconf{G'}{L'}{\nameN'}{\regk}$
\end{enumerate}
\end{theorem}
\begin{proof}
We show something stronger, whereby reduction is not necessarily to $\regk$ but to some $S_{ret}$ taken from:
  $
S_{ret}\ ::=\ v\mid \tuple{\vec v}\mid \mathbb{M}
  $.

  The proof of the forward direction is by structural induction on the derivation of the reduction (\cref{appendix:entailment:1})
    and the reverse is by induction on the number of steps in the reduction sequence
(\cref{appendix:entailment:2}). 
Here we show the case of function application.
  We refer to \cref{lem:congruence} by (CG).

\noindent
\textbf{Forward direction.}
Let $S = \Efuncall{f}{M_1 , .. , M_n}$ and suppose
$\nameN(f) = (\vec y,\vec z,S_{b})$ and
$L_f = \{\vec y \mapsto \vec v\} \uplus \{\{\vec z\} \mapsto 0\}$.  
We assume~(1):
\[
  \irule[FunCall][funcall]{
    \bsconf{G}{L}{\nameN}{M_n}
    \bsevalexp
    \bsconf{G_1}{L_1}{\nameN}{v_n}\\ \dots\
    \bsconf{G_{n-1}}{L_{n-1}}{\nameN}{M_1}
    \bsevalexp
    \bsconf{G_n}{L_n}{\nameN}{v_1}
    \\
    \bsconf{G_n}{L_f}{\nameN}{{S_b}}
    \bseval
    \bsconf{G'}{L'}{\nameN}{\modeM}
  }{
    \bsconf{G}{L}{\nameN}{\Efuncall{f}{M_1 , .. , M_n}}
    \bsevalexp
    \bsconf{{G'}}{L_n}{\nameN}{S_{ret}}
  }
\]
with $S_{ret}$ being: $L'(z_1)$ if $ |\vec z|=1$;
$\tuple{L'(z_1),..,L'(z_m)}$ if $|\vec z| > 1$;
$\regk$                        if $|\vec z| = 0$.
We then obtain~(2) as follows:
\begin{align*}
&\bsconf{G}{L}{\nameN}{\Efuncall{f}{M_1 , .. , M_n}} \\
&\to^*\bsconf{G_1}{L_1}{\nameN}{\Efuncall{f}{M_1 , .. M_{n-1}, v_n}}
&&\quad\text{by CG and IH$\bsconf{G}{L}{\nameN}{M_n}$}\\
&\dots\\
&\to^*\bsconf{G_{n}}{L_{n}}{\nameN}{\Efuncall{f}{v_1 , .. , v_n}}
&&\quad\text{by CG and IH$\bsconf{G_{n-1}}{L_{n-1}}{\nameN}{M_1}$}\\
&\to\bsconf{G_{n}}{L_f}{\nameN}{\ssframe{S_b}^{\vec z}_{L}}
&&\quad\text{where $\nameN(f),L_f$ as above}\\
&\to^*\bsconf{G'}{L'}{\nameN}{\ssframe{\modeM}^{\vec z}_{L_n}}
&&\quad\text{by CG and IH$\bsconf{G_{n}}{L_f}{\nameN}{S_b}$}\\
&\to\bsconf{G'}{L_n}{\nameN}{S_{ret}}
\end{align*}

\noindent
\textbf{Reverse direction.}
For a term $S = \Efuncall{f}{M_1, .., M_n}$, an evaluation to $S_{ret}$ as in~(2) requires a trace:
\begin{align*}
&\bsconf{G_0}{L_0}{\nameN}{\Efuncall{f}{M_1 , .. , M_n}}
\to^*
\bsconf{G_{n}}{L_{n}}{\nameN}{\Efuncall{f}{v_1 , .. , v_n}}\\
&
\to^*
\bsconf{G_{n}}{L_f}
{\nameN}
  {\ssframe{S_b}^{\vec z}_{L_n}}
\to^*
\bsconf
{G'}
{L''}
{\nameN}
{\ssframe{\modeM}^{\vec z}_{L_n}}
\to^*
\bsconf
{G'}
{L_n}
{\nameN}
{S_{ret}}
\end{align*}
with $\nameN(x),L_f,S_{ret}$ satisfying the same conditions as in the forward direction above, and with intermediate steps:
\begin{align*}
  &\bsconf{G_{n-i}}{L_{n-i}}{\nameN}{M_i}
  \to^*
  \bsconf{G_{n-i+1}}{L_{n-i+1}}{\nameN}{v_i} \tag{A$_i$}
  \\&
  \bsconf{G_{n}}
  {L_f}
  {\nameN}
  {S_b}
  \to^*
  \bsconf{G'}
  {L'}
  {\nameN}
  {\modeM}
  \tag{B}
\end{align*}
To obtain~(1), by IH on~($A_i$):
$
\bsconf{G_{n-i}}{L_{n-i}}{\nameN}{M_i}
\bseval
\bsconf{G_{n-i+1}}{L_{n-i+1}}{\nameN}{v_i}
$.
\\
Next, by IH on~(B):
$
\bsconf{G_{n}}
{L_f}
{\nameN}
{S_b}
\bseval
\bsconf{G'}{L'}{\nameN}{\modeM}
$.
We therefore obtain by \iref{funcall}:
$
\bsconf{G_0}{L_0}{\nameN}{\Efuncall{f}{M_1 , .. , M_n}}
\bsevalexp
\bsconf{G'}{L_n}{\nameN}{S_{ret}}
$.
~\qed
\end{proof}

%


  \FloatBarrier
  \section{Implementation}
  \label{sec:imp}
  We {implemented} \yult{}~\cite{tools:yul:yultracer, tools:yul:release}, a prototype interpreter {based on our operational semantics}
for Yul that is parametric on dialect implementations. Currently, a partial implementation of the Shanghai {upgrade} (Solidity 0.8.20 to 0.8.23) of the EVM execution specifications~\cite{evm-execution-specs} is provided as its sole dialect. We provide opcodes for the following categories: Arithmetic, Bitwise, Comparison, Control-flow, Memory, Stack, and Storage. We omit the implementation of opcodes that are not listed in the official documentation for Yul~\cite{yul} under version 0.8.23, but otherwise adapt the Python specification to our use as faithfully as possible. \yult implements our small-step semantics in the style of CEK machines~\cite{FelleisenF87} and is written using OCaml. Modularity for dialects is achieved using functors: a module signature is provided for dialects to satisfy, which are then used to instantiate any concrete implementation of the interpreter. Lastly, we also provide in the EVM dialect a partial implementation of gas consumption. Since Yul replaces some operations that consume gas but does not provide an official model for gas consumption, we do not calculate gas consumption with perfect accuracy, but instead attempt to accurately compute it for the subset of opcodes provided.

We evaluate our interpreter using custom tests\,---1263 Lines of Code (LoC)---\,and a subset of tests found in the Solidity {compiler} repository~\cite{solidity-repo}\,---269 LoC. From the 50 files found under \textsf{test/libyul/yulInterpreterTests} in the Solidity repository (commit 2b2c76c), we kept 19 files that do not make use of unimplemented EVM instructions: opcodes for inter-contract communication in the System category; those in the Block, Environment or Log categories; the opcode for Keccak, which is implemented but unused as it uses a model that requires symbolic execution; or opcodes implemented in the Cancun upgrade (i.e. Solidity 0.8.24) such as those for transient memory or \textsf{mcopy}. From the 19 remaining tests, we mark and omit 3 tests that do not behave according to the official specifications of the EVM; these are: \textsf{access\_large\_memory\_offsets.yul}, which under standard operation would exceed the gas limit when trying to access large offsets but does not; \textsf{expr\_nesting\_depth\_exceeded.yul}, which imposes a nesting limit on expressions that is not specified in the Yul documentation; and \textsf{hex\_literals.yul}, which {uses} underscores \textsf{\_} in hexadecimal strings, not allowed by the official Yul grammar. For our own custom examples, we implemented two categories of programs: sequence generators (426 LoC) and sorting algorithms (837 LoC). For the former, we have generators for Catalan numbers, the Heighway dragon curve, Fibonacci numbers, Pell numbers, Prime numbers, and the Thue–Morse sequence. For the latter we have implementations of bubble, heap, insertion, quick and shell sort on arrays of 6, 300 and 1000 items. These were chosen as non-trivial algorithms with standard implementations that we can compare to for reference that also make use of a range of Yul features and EVM opcodes while testing {\yult's} performance.

All examples were executed on a laptop featuring an 8-core Intel Core i7-8665U CPU with 32 GiB of RAM running Ubuntu 23.04. The interpreter was compiled with the OCaml 4.10.0 compiler using the Dune build system. We measure for each batch of examples a three-trial average total execution time of 0.478s for all 8 sequence generator files, 31.960s for all 12 sorting algorithm files, and 2.932s for all 16 of the Solidity {tests} (including an infinite loop set to time out at 2s).
{\yult produced the expected output for all programs.}


  \section{Related Work}
  \label{sec:relwork}
As discussed in the Introduction, the body of work on formalising Yul involves mechanisations in e.g.\ Dafny~\cite{mechanisation:yul:dafny,mechanisation:yul:dafny:article}, the K-framework \cite{mechanisation:yul:K}, Isabelle/HOL \cite{mechanisation:yul:isabelle}, Lean \cite{mechanisation:yul:lean}, and ACL2 \cite{mechanisation:yul:acl2}. To our knowledge, there is no peer-reviewed tool-independent presentation of a formal semantics for Yul. Closely related is the area of formalising Solidity, for which
a more substantial body of work exists: e.g.\ \cite{related:solidity:SSCalc}~mathematically presents a calculus to reason about Solidity smart contracts and then implements and mechanically verifies its soundness in Isabelle/HOL; \cite{related:solidity:fstar}~presents a shallow embedding of Solidity in F*; \cite{related:solidity:towards-opsem} presents an operational semantics for a subset of Solidity and presents an abstract model for memory; \cite{related:solidity:towards-ver}~formalises another subset of Solidity by presenting a big-step semantics; \cite{related:solidity:smt-memory}~presents a semantics for the memory model of Solidity by mapping a fragment of Solidity to SMT-based constructs; \cite{related:solidity:denot} implements a denotational semantics for a subset of Solidity in 1500 lines of Isabelle. Also relevant are specifications for the EVM, which are necessary to implement the EVM dialect for Yul, such as: the Ethereum execution specification in Python~\cite{evm-execution-specs}; other formalisation work such as~\cite{related:evm:fstar}, which presents the first complete small-step semantics of EVM bytecode and formalises it in F*; and~\cite{related:evm:dafny}, which presents a complete and formal specification in Dafny. Lastly, formalisations efforts are also common practice outside smart contracts, such as with the Glasgow Haskell Compiler (GHC) intermediate language {Core} which implements System FC~\cite{related:system-fc} and features a maintained formal specification~\cite{related:ghc:core}.

\section{Conclusions}
  \label{sec:conclusion}
In this paper, we have presented an operational semantics for Yul.
 {We first defined a source-level abstract syntax constructed by inspecting the grammar specified in \cite{yul} (\Cref{sec:syntax}), which we subsequently extended
   with constructs necessary for evaluation (\Cref{sec:sem})}. We then {presented} big- and small-step semantics, constructed by carefully inspecting the pseudocode evaluation function provided in \cite{yul}, and supplied a proof of semantic equivalence between them (\Cref{theorem:equiv}). Lastly,
 we presented \yult, an interpreter that {implements our small-step semantics
   for} a subset of the EVM dialect based on the Shanghai {upgrade}~\cite{evm-execution-specs}. We evaluated our implementation against tests from the Solidity repository~\cite{solidity-repo} and our own ones based on standard sequence generators and sorting algorithms.

 We believe our {work} offers multiple advantages. {Being} in standard mathematical notation, {and free of implementation or formalisation considerations,}
{our semantics can} serve to unify various extant tool-specific formalisations and verify any current and future implementations,
 {while it can also be directly used} 
 in mathematical proofs, especially those involving the development and verification of language properties. {Moreover}, we provide an exhaustive specification for statements that irons out issues with ambiguity arising from under-specification of the semantics; e.g. function returns with no return variables. Lastly, we anticipate that a semantics such as ours will aid the development of future theories for Yul. We shall discuss this last point in more detail next.

This document focuses only the semantics of Yul statements. As such, we leave a detailed discussion of the EVM dialect and object semantics for future {work}. In particular, the EVM semantics includes behaviours which may be of technical interest,  such as instructions related to transactions and interaction with the blockchain, as well as inter-contract communication. {\yult{} and our semantics work} towards the development of symbolic execution {tools and theories developed directly in Yul, avoiding low-level complexities of EVM bytecode and the feature-rich, evolving syntax of Solidity};
in particular those which combine symbolic execution with game semantics (e.g. \cite{LinT20} and \cite{KoutavasLT22}) to {model} higher-order interactions of smart contracts {relevant to} reentrancy exploits.
Lastly, {our} semantics {could lead to} soundness proofs for a {future Yul} type system.


  \small{\paragraph{Acknowledgements:}
  This work was supported in part by: Science Foundation Ireland under grant number 13/RC/2094\_2;
  the Cisco University Research Program Fund, a corporate advised fund of Silicon Valley Community Foundation; and Ethereum Foundation grant FY23-1127.
  For the purpose of
  Open Access, the authors have applied a CC BY public copyright licence to any Author Accepted
  Manuscript version arising from this submission.}


  \bibliographystyle{splncs04}

  \bibliography{references}

\begin{thebibliography}{10}
\providecommand{\url}[1]{\texttt{#1}}
\providecommand{\urlprefix}{URL }
\providecommand{\doi}[1]{https://doi.org/#1}

\bibitem{eth-vulnerabilities}
Smart contract vulnerabilities ({SCV}) list,
  \url{https://github.com/sirhashalot/SCV-List}, accessed 6 June 2024

\bibitem{tools:bytecode:oyente}
{Oyente} (Mar 2017),
  \url{https://github.com/enzymefinance/oyente/releases/tag/0.2.7}

\bibitem{mechanisation:yul:isabelle}
{Yul-Isabelle: Executable Formal Semantics of Yul} (2021),
  \url{https://github.com/mmalvarez/Yul-Isabelle/tree/master}, accessed 24 June
  2024

\bibitem{tools:bytecode:manticore}
{Manticore} (Feb 2022),
  \url{https://github.com/trailofbits/manticore/releases/tag/0.3.7}

\bibitem{tools:solidity:echidna}
{Echidna} (Mar 2024),
  \url{https://github.com/crytic/echidna/releases/tag/v2.2.3}

\bibitem{tools:solidity:gambit}
{Gambit: Mutant Generation for Solidity} (May 2024),
  \url{https://github.com/Certora/gambit/releases/tag/v1.0.5}

\bibitem{tools:bytecode:mythril}
{Mythril} (Mar 2024),
  \url{https://github.com/Consensys/mythril/releases/tag/v0.24.8}

\bibitem{tools:solidity:slither}
{Slither} (Jun 2024),
  \url{https://github.com/crytic/slither/releases/tag/0.10.3}

\bibitem{tools:solidity:solhint}
{solhint} (May 2024),
  \url{https://github.com/protofire/solhint/releases/tag/v5.0.1}

\bibitem{related:solidity:fstar}
Bhargavan, K., Delignat-Lavaud, A., Fournet, C., Gollamudi, A., Gonthier, G.,
  Kobeissi, N., Kulatova, N., Rastogi, A., Sibut-Pinote, T., Swamy, N.,
  Zanella-B\'{e}guelin, S.: Formal verification of smart contracts: Short
  paper. In: Proceedings of the 2016 ACM Workshop on Programming Languages and
  Analysis for Security. p. 91–96. PLAS '16, Association for Computing
  Machinery, New York, NY, USA (2016). \doi{10.1145/2993600.2993611}

\bibitem{mechanisation:yul:dafny:article}
Cassez, F.: {Formal and executable semantics of Yul in Dafny} (oct 2023),
  \url{https://hackmd.io/@FranckC/BJz02K4Za}, accessed 24 June 2024

\bibitem{mechanisation:yul:dafny}
Cassez, F.: {yul-dafny} (oct 2023),
  \url{https://github.com/franck44/yul-dafny}, accessed 24 June 2024

\bibitem{related:evm:dafny}
Cassez, F., Fuller, J., Ghale, M.K., Pearce, D.J., Quiles, H.M.A.: Formal
  and executable semantics of the ethereum virtual machine in dafny. In:
  Chechik, M., Katoen, J.P., Leucker, M. (eds.) Formal Methods. pp. 571--583.
  Springer (2023)

\bibitem{related:solidity:towards-opsem}
Crosara, M., Centurino, G., Arceri, V.: Towards an operational semantics for
  solidity (2019), \url{https://api.semanticscholar.org/CorpusID:249983842}

\bibitem{crypto-crimes-2024}
{Crystal Intelligence}: Adolescent anarchy: Thirteen years of crypto crimes
  unveiled,
  \url{https://crystalintelligence.com/rohirov/2024/06/Crystal-Intelligence-Thirteen-Years-of-Crypto-Crimes-Unveiled.pdf},
  accessed 25 June 2024

\bibitem{evm-execution-specs}
{ethereum.org}: Ethereum execution client specifications,
  \url{https://github.com/ethereum/execution-specs.git}, accessed 14 June 2024

\bibitem{solidity-repo}
{ethereum.org}: The {S}olidity contract-oriented programming language,
  \url{https://github.com/ethereum/solidity.git}, accessed 13 June 2024

\bibitem{mechanisation:yul:K}
{ethereum.org}: {Yul-K} (2019), \url{https://github.com/ethereum/Yul-K},
  accessed 24 June 2024

\bibitem{tools:solidity:slither:paper}
Feist, J., Greico, G., Groce, A.: Slither: A static analysis framework for
  smart contracts (05 2019). \doi{10.1109/WETSEB.2019.00008}

\bibitem{FelleisenF87}
Felleisen, M., Friedman, D.P.: Control operators, the secd-machine, and the
  {\(\lambda\)}-calculus. In: Wirsing, M. (ed.) Formal Description of
  Programming Concepts - {III:} Proceedings of the {IFIP} {TC} 2/WG 2.2 Working
  Conference on Formal Description of Programming Concepts - III, Ebberup,
  Denmark, 25-28 August 1986. pp. 193--222. North-Holland (1987)

\bibitem{Smaragdakis-decompiler}
Grech, N., Lagouvardos, S., Tsatiris, I., Smaragdakis, Y.: Elipmoc: advanced
  decompilation of ethereum smart contracts. Proc. ACM Program. Lang.
  \textbf{6}(OOPSLA1) (apr 2022). \doi{10.1145/3527321}

\bibitem{tools:solidity:echidna:paper}
Grieco, G., Song, W., Cygan, A., Feist, J., Groce, A.: Echidna: effective,
  usable, and fast fuzzing for smart contracts. pp. 557--560 (07 2020).
  \doi{10.1145/3395363.3404366}

\bibitem{related:evm:fstar}
Grishchenko, I., Maffei, M., Schneidewind, C.: A semantic framework for the
  security analysis of ethereum smart contracts. In: Bauer, L., K{\"u}sters, R.
  (eds.) Principles of Security and Trust. pp. 243--269. Springer (2018)

\bibitem{related:solidity:smt-memory}
Hajdu, {\'A}., Jovanovi{\'{c}}, D.: {SMT}-friendly formalization of the
  {S}olidity memory model. In: M{\"u}ller, P. (ed.) Programming Languages and
  Systems. pp. 224--250. Springer (2020)

\bibitem{mechanisation:yul:acl2}
{Kestrel Institute}: {Yul Directory, Kestrel Books in ACL2 Repository} (2023),
  \url{https://github.com/acl2/acl2/tree/master/books/kestrel/yul}, accessed 24
  June 2024

\bibitem{KoutavasLT22}
Koutavas, V., Lin, Y., Tzevelekos, N.: From bounded checking to verification of
  equivalence via symbolic up-to techniques. In: TACAS. LNCS, vol. 13244, pp.
  178--195. Springer (2022). \doi{10.1007/978-3-030-99527-0\_10}

\bibitem{tools:yul:release}
Lin, Y.Y.: {LaifsV1/YulTracer: Alpha 0.1.1: Parametric Parser and Main} (Jun
  2024). \doi{10.5281/zenodo.12511493},
  \url{https://doi.org/10.5281/zenodo.12511493}

\bibitem{tools:yul:yultracer}
Lin, Y.Y.: Yultracer: Alpha 0.1.1 (artefact) (Jun 2024).
  \doi{10.5281/zenodo.12588076}, \url{https://doi.org/10.5281/zenodo.12588076}

\bibitem{LinT20}
Lin, Y., Tzevelekos, N.: Symbolic execution game semantics. In: {FSCD}. Schloss
  Dagstuhl - Leibniz-Zentrum f{\"{u}}r Informatik (2020)

\bibitem{tools:bytecode:oyente:paper}
Luu, L., Chu, D., Olickel, H., Saxena, P., Hobor, A.: Making smart contracts
  smarter. In: Weippl, E.R., Katzenbeisser, S., Kruegel, C., Myers, A.C.,
  Halevi, S. (eds.) Proceedings of the 2016 {ACM} {SIGSAC} Conference on
  Computer and Communications Security, Vienna, Austria, October 24-28, 2016.
  pp. 254--269. {ACM} (2016). \doi{10.1145/2976749.2978309}

\bibitem{related:solidity:denot}
Marmsoler, D., Brucker, A.D.: A denotational semantics of {S}olidity in
  {I}sabelle/{HOL}. In: Calinescu, R., P{\u{a}}s{\u{a}}reanu, C.S. (eds.)
  Software Engineering and Formal Methods. pp. 403--422. Springer (2021)

\bibitem{related:solidity:SSCalc}
Marmsoler, D., Thornton, B.: Sscalc: A calculus for solidity smart contracts.
  In: Ferreira, C., Willemse, T.A.C. (eds.) Software Engineering and Formal
  Methods. pp. 184--204. Springer (2023)

\bibitem{tools:bytecode:manticore:paper}
Mossberg, M., Manzano, F., Hennenfent, E., Groce, A., Grieco, G., Feist, J.,
  Brunson, T., Dinaburg, A.: Manticore: A user-friendly symbolic execution
  framework for binaries and smart contracts. pp. 1186--1189 (11 2019).
  \doi{10.1109/ASE.2019.00133}

\bibitem{mechanisation:yul:lean}
{NethermindEth}: {Yul IR specification in Lean} (2021),
  \url{https://github.com/NethermindEth/Yul-Specification}, accessed 24 June
  2024

\bibitem{related:ghc:core}
{Peyton Jones}, S.: System {FC}, as implemented in {GHC} (May),
  \url{https://gitlab.haskell.org/ghc/ghc/-/blob/master/docs/core-spec/core-spec.pdf},
  accessed 27 June 2024

\bibitem{tools:bytecode:mythril:survey}
Sharma, N., Sharma, S.: A survey of mythril, a smart contract security analysis
  tool for evm bytecode  \textbf{13},  51003--51010 (12 2022)

\bibitem{eth-bugs}
{Solidity Team}: List of known bugs - {S}olidity 0.8.27 documentation,
  \url{https://docs.soliditylang.org/en/develop/bugs.html}, accessed 6 June
  2024

\bibitem{yul}
{Solidity Team}: {Y}ul - {S}olidity 0.8.27 documentation,
  \url{https://docs.soliditylang.org/en/latest/yul.html}, accessed 6 June 2024

\bibitem{related:system-fc}
Sulzmann, M., Chakravarty, M., Peyton~Jones, S., Donnelly, K.: System f with
  type equality coercions. In: ACM SIGPLAN International Workshop on Types in
  Language Design and Implementation (TLDI'07). pp. 53--66. ACM (January 2007),
  \url{https://www.microsoft.com/en-us/research/publication/system-f-with-type-equality-coercions/}

\bibitem{related:solidity:towards-ver}
Zakrzewski, J.: Towards verification of ethereum smart contracts: A
  formalization of core of solidity. In: Piskac, R., R{\"u}mmer, P. (eds.)
  Verified Software. Theories, Tools, and Experiments. pp. 229--247. Springer
  (2018)

\end{thebibliography}

  \clearpage
  \appendix

  \section{Restrictions for Sub-Statements}
  In this section we define and prove properties of our semantics that will become useful to define syntactic restrictions to the source-level language.

\begin{definition}[Sub-statements]\label{def:sub}
We say a statement $S$ is a direct sub-statement of $S'$, written $S \prec S'$, if:
\begin{itemize}
\item $S' = \Sblock{S_1 .. S_n}$ such that $S \in \{S_1, .., S_n\}$
\item $S' = \Scond{M}{S}$
\item $S' = \Sswitch{M}{\casek\,v_1\,S_1 .. \casek\,v_n\,S_n}{\defaultk S_d}$ such that $S = S_d$ or $S \in \{S_1, .., S_n\}$
\item $S' = \Sfor{S_i}{M}{S_p}{S_b}$ such that $S \in \{S_i, S_p, S_b\}$
\item $S' = \Sfundefto{x}{\vec y}{\vec z}{S}$
\end{itemize}
We also say a statement $S$ is a transitive sub-statement of $S'$, written $S \preceq S'$ if:
\begin{itemize}
\item $S' = S$
\item $S' = \Sblock{S_1 .. S_n}$ such that $\exists i \in [1 , n] . S \preceq S_i$
\item $S' = \Scond{M}{S_b}$ if $S \preceq S_b$
\item $S' = \Sswitch{M}{\casek\,v_1\,S_1 .. \casek\,v_n\,S_n}{\defaultk S_d}$ such that $\exists i \in [1 , n] . S \preceq S_i$
\item $S' = \Sfor{S_i}{M}{S_p}{S_b}$ such that $\exists S'' \in \{S_i, S_p, S_b\} . S \preceq S''$
\item $S' = \Sfundefto{x}{\vec y}{\vec z}{S_b}$ and $S \preceq S_b$
\end{itemize}
When generally referring to sub-statements without specifying which kind, we always mean transitively.
\end{definition}

\begin{definition}[Syntactic restrictions for halting statements]\label{def:halt:restrict}
We impose the following syntactic restrictions on halting statements:
\begin{itemize}
\item \textbf{Loop halting:} Let $S \in \{\breakk , \continuek\}$.
For all $S \preceq S'$, there exists an $S_f = \Sfor{\Sblock{S_i}}{M}{{ S_{p}}}{{ S_{b}}}$ such that $S_f \preceq S'$ and $S \preceq S_b$, but there cannot be a statement $S_f' = \Sfundefto{x}{\vec y}{\vec z}{S_b'}$ such that $S_f' \preceq S_f$ and $S \preceq S_b'$.
\item \textbf{Function halting:} For all $\leavek \preceq S'$, there exists an $S_f = \Sfundefto{x}{\vec y}{\vec z}{S_b}$ such that $S_f \preceq S'$ and $S \preceq S_b$.
\end{itemize}
\end{definition}

\begin{lemma}[Commutativity of sequencing]\label{lem:seq:comm}
Given two non-empty sequences of statements $\vec S$ and $\vec S'$, then (1) $\iff$ (2) where:
\begin{flalign*}
\text{(1) } &
\bsconf{G}{L}{\nameN}{\vec S , \vec S'}
\bseval
\bsconf{G''}{L''}{\nameN''}{\modeM}
\\
\begin{aligned}
\text{(2) }
\\\phantom{}
\end{aligned}
&
\begin{aligned}
\exists (G',L',\nameN') .
&\bsconf{G}{L}{\nameN}{\vec S}
\bseval
\bsconf{G'}{L'}{\nameN'}{\regk}
\\
&\land
\bsconf{G'}{L'}{\nameN'}{\vec S'}
\bseval
\bsconf{G''}{L''}{\nameN''}{\modeM}
\end{aligned}
\end{flalign*}
\end{lemma}
\begin{proof}
This proof follows by structural induction on the derivation tree of $\bsconf{G}{L}{\nameN}{\vec S}$ and by case analysis on the form of $\vec S$. The base case, where $\vec S = S_1$, follows by direct application of \iref{seqreg}. The inductive case where $\vec S = S_1 , \vec S_n$ follows by application of \iref{seqreg} and by the inductive hypothesis on the derived configurations containing $\vec S_n$. By assumption, \iref{seqirreg} does not apply because $\bsconf{G}{L}{\nameN}{\vec S}$ terminates regularly, while \iref{seqempty} holds syntactically.
\end{proof}

\begin{lemma}[Containment of Loop Halting Statements]\label{lem:break:cont}
Given a configuration $\bsconf{G}{L}{\nameN}{\Sfor{S_i}{M}{{ S_{p}}}{{ S_{b}}}}$ and $\modeM \in \{\breakk,\continuek\}$, then:
\[
      \bsconf{G}{L}{\nameN}{\Sfor{S_i}{M}{{ S_{p}}}{{ S_{b}}}}
      \not\bseval
      \bsconf{G'}{L'}{\nameN}{\modeM}
\]
\end{lemma}
\begin{proof}
This follows by induction on the derivation tree of $\bsconf{G}{L}{\nameN}{\Sfor{S_i}{M}{{ S_{p}}}{{ S_{b}}}}$ and by inspection on the inference rules for loops. Base cases:
\begin{itemize}
\item \iref{forfalse} always evaluates to \regk;
\item \iref{forhalt1} catches all \breakk in $S_b$ and produces \regk as the top-level result;
\item \iref{forhalt2} only handles \leavek and thus cannot evaluate to \breakk or \continuek.
\end{itemize}
Inductive cases:
\begin{itemize}
\item \iref{forloop} follows by direct application of the inductive hypothesis -- \continuek is not propagated because the result of evaluating $S_b$ is thrown away;
\item \iref{forinit} derives a statement prefixed with $S_i$, which by Definition~\ref{def:halt:restrict} does not contain any statements that evaluate to \breakk or \continuek, so by Lemma~\ref{lem:seq:comm} we directly apply the inductive hypothesis.
\end{itemize}
\qed
\end{proof}

\begin{lemma}[Lower bound of $L$]\label{lem:bound:l}
Given a configuration $\bsconf{G}{L}{\nameN}{S}$ for some $S$ in source-level syntax, if:
\[
      \bsconf{G}{L}{\nameN}{S}
      \to^*
      \bsconf{G'}{L'}{\nameN'}{\modeM}
\]
then $\dom{L'} \supseteq \dom{L}$.
\end{lemma}
\begin{proof}
This follows by structural induction and inspection on the form of $S$.
\qed
\end{proof}

\begin{lemma}[Invariance of $\dom{L}$ in loops]\label{lem:l:invariance:loops}
Given a configuration $\bsconf{G}{L}{\nameN}{\Sfor{S_i}{M}{{ S_{p}}}{{ S_{b}}}}$, if:
\[
      \bsconf{G}{L}{\nameN}{\Sfor{S_i}{M}{{ S_{p}}}{{ S_{b}}}}
      \to^*
      \bsconf{G'}{L'}{\nameN}{\modeM}
\]
then $\dom{L} = \dom{L'}$.
\end{lemma}
\begin{proof}
This follows by structural induction and inspection on the reduction rules. Consider loop initialisation:
\begin{align*}
  \ssconf{\Sfor{\Sblock{S_1 .. S_n}}{M}{{S_{p}}}{{ S_{b}}}}{L}{\nameN}
  &
  \to
  \ssconf{
        \Sblock{S_1 .. S_n\,\Sfor{\Sblock{}}{M}{S_{p}}{ S_{b}}}
  }{L}{\nameN}\\
  &
  \to
  \ssconf{
        \Sblock{S_1 .. S_n\,\Sfor{\Sblock{}}{M}{S_{p}}{ S_{b}}}_L^\nameN
  }{L}{\nameN}
\end{align*}
By case analysis on the rules for blocks:
\begin{itemize}
\item if $\ssconf{S_1 .. S_n}{L}{\nameN} \to^* \ssconf{\regk}{L'}{\nameN'}$, then
\begin{align*}
  \to^*
  \ssconf{
        \Sblock{\regk\,\Sfor{\Sblock{}}{M}{S_{p}}{ S_{b}}}_L^\nameN
  }{L'}{\nameN'}
  &\to
  \ssconf{
        \Sblock{\Sfor{\Sblock{}}{M}{S_{p}}{ S_{b}}}_L^\nameN
  }{L'}{\nameN'}\\
  &\to^*
  \ssconf{
        \Sblock{\modeM}_L^\nameN
  }{L''}{\nameN'}\\
    &\to^*
    \ssconf{
          \modeM
    }{L''\restrict L}{\nameN}
\end{align*}
We know $\dom{L} = \dom{L'' \restrict L}$ by Lemma~\ref{lem:bound:l}.

\item if $\ssconf{S_1 .. S_n}{L}{\nameN} \to^* \ssconf{\modeM}{L'}{\nameN'}$ where $\modeM \neq \regk$, this case is simpler than the previous case and shown by inspection.
\end{itemize}
The remaining cases are simpler than the one for initialisation.
\qed
\end{proof}


  \section{Big-Step and Small-Step Equivalence}
\label{apx:equiv}
\begin{lemma}[Semantic Entailment 1]\label{lem:1}
Given statement $S$ in source syntax, global environment $G$, state $L$, namespace $\nameN$, 
then for a configuration $\bsconf{G}{L}{\nameN}{S}$:
\begin{enumerate}[ (1) ]
\item $\bsconf{G}{L}{\nameN}{S} \bseval \bsconf{G'}{L'}{\nameN'}{S_{ret}}$
\item $\bsconf{G}{L}{\nameN}{S} \to^* \bsconf{G'}{L'}{\nameN'}{S_{ret}}$
\end{enumerate}
it is the case that (1) $\implies$ (2).
\end{lemma}
\subsection{Proof of Semantic Entailment 1}\label{appendix:entailment:1}
Consider a configuration $\bsconf{G}{L}{\nameN}{S}$ such that
\[\bsconf{G}{L}{\nameN}{S} \bseval \bsconf{G'}{L'}{\nameN'}{S_{r}}\]

By structural induction on the derivation tree of the big-step evaluation of $\bsconf{G}{L}{\nameN}{S}$ and by case analysis on the form of $S$, we have base cases:
\begin{itemize}
\item \textbf{Empty Block:}
Let $S = \Sblock{\varepsilon}$, then for (1) we have
\[
    \irule[Block]{
    \irule[SeqEmpty]{
    }{
      \bsconf{G}{L}{\nameN}{\varepsilon}
      \bsevalseq
      \bsconf{G}{L}{\nameN}{\regk}
    }
    }{
      \bsconf{G}{L}{\nameN}{\Sblock{}}
      \bseval
      \bsconf{G}{L}{\nameN}{\regk}
    }
\]
For (2), we have by definition:
\[
\bsconf{G}{L}{\nameN}{\Sblock{}} \to \bsconf{G}{L}{\nameN}{\regk}
\]
Therefore (1) $=$ (2).
\item \textbf{Function Definition:}
Let $S = \Sfundefto{x}{\vec y}{\vec z}{M}$, then for (1) we have
\[
    \irule[FunDef]{
    }{
      \bsconf{G}{L}{\nameN}{\Sfundefto{x}{\vec y}{\vec z}{M}}
      \bseval
      \bsconf{G}{L}{\nameN}{\regk}
    }
\]
For (2), we have by definition:
\[
\bsconf{G}{L}{\nameN}{\Sfundefto{x}{\vec y}{\vec z}{M}} \to \bsconf{G}{L}{\nameN}{\regk}
\]
Therefore (1) $=$ (2).
\item \textbf{Identifier:}
Let $S = x$, then for (1) we have
\[
    \irule[Ident]{
    v = L(x)
    }{
      \bsconf{G}{L}{\nameN}{x}
      \bsevalexp
      \bsconf{G}{L}{\nameN}{v}
    }
\]
For (2), we have by definition:
\[
\bsconf{G}{L}{\nameN}{x} \to \bsconf{G}{L}{\nameN}{v}\qquad\text{where $v = L(x)$}
\]
Therefore (1) $=$ (2).
\end{itemize}

With base cases done, we now proceed with the inductive cases. Recall from structural induction that the inductive hypothesis (IH) states that:
\begin{align*}
&\bsconf{G'}{L'}{\nameN'}{S'} \bseval \bsconf{G''}{L''}{\nameN''}{S_{ret}}
\\
&\implies \bsconf{G'}{L'}{\nameN'}{S'} \to^* \bsconf{G''}{L''}{\nameN''}{S_{ret}}
\end{align*}
for any $\bsconf{G'}{L'}{\nameN'}{S'} \bseval \bsconf{G''}{L''}{\nameN''}{S_{ret}}$ derivable from $\bsconf{G}{L}{\nameN}{S}$. Also recall Lemma~\ref{lem:congruence} which states congruence (CG).

\begin{itemize}
\item \textbf{Blocks:}
Let $S = \Sblock{S_1 .. S_n}$, $\vec S_n = S_2, .. , S_n$ and $\nameN' = \nameN \uplus \funsof{\Sblock{S_1 .. S_n}}$. By case analysis on the evaluation of $S_1$, for (1) we have the following cases:
\begin{enumerate}[{[}A{]}]
\item \textbf{Regular Evaluation:}
\[
    \irule[Block]{
    \irule[SeqReg]{
      \bsconf{G}{L}{\nameN'}{S_1}
      \bsevalseq
      \bsconf{G_1}{L_1}{\nameN'}{\regk}
      \\\\
      \bsconf{G_1}{L_1}{\nameN'}{\vec S_n}
      \bsevalseq
      \bsconf{G_n}{L_n}{\nameN'}{\modeM_n}
    }{
      \bsconf{G}{L}{\nameN'}{S_1 , \vec S_n}
      \bsevalseq
      \bsconf{G_n}{L_n}{\nameN'}{\modeM_n}
    }
    }{
      \bsconf{G}{L}{\nameN}{\Sblock{S_1 .. S_n}}
      \bseval
      \bsconf{G_n}{L_n\restrict L}{\nameN}{\modeM_n}
    }
\]

\item \textbf{Irregular Evaluation:}
\[
    \irule[Block]{
    \irule[SeqIrreg]{
      \bsconf{G}{L}{\nameN'}{S_1}
      \bseval
      \bsconf{G_1}{L_1}{\nameN'}{\modeM_1}
      \\
      \modeM_1 \neq \regk
    }{
      \bsconf{G}{L}{\nameN'}{S_1 , \vec S_n}
      \bsevalseq
      \bsconf{G_1}{L_1}{\nameN'}{\modeM_1}
    }
    }{
      \bsconf{G}{L}{\nameN}{\Sblock{S_1 .. S_n}}
      \bseval
      \bsconf{G_1}{L_1\restrict L}{\nameN}{\modeM_1}
    }
\]
\end{enumerate}
Now, for (2) we have the following cases:
\begin{enumerate}[(A)]
\item Assume [A] is the case. We have the following trace:
\begin{align*}
&\bsconf{G}{L}{\nameN}{\Sblock{S_1 .. S_n}}
\\
&\to
\bsconf{G}{L}{\nameN'}{\Sblock{S_1,\vec S_n}_L^\nameN}
\\
&\to^*
\bsconf{G_1}{L_1}{\nameN'}{\Sblock{\regk,\vec S_n}_L^\nameN}
&&\text{by CG and IH$\bsconf{G}{L}{\nameN'}{S_1}$ in [A]}
\\
&\to
\bsconf{G_1}{L_1}{\nameN'}{\Sblock{\vec S_n}_L^\nameN}
\\
&\to^*
\bsconf{G_n}{L_n}{\nameN'}{\Sblock{\modeM_n}_L^\nameN}
&&\text{by CG and IH$\bsconf{G_1}{L_1}{\nameN'}{\vec S_n}$ in [A]}
\\
&\to
\bsconf{G_n}{L_n\restrict L}{\nameN}{{\modeM_n}}
\end{align*}
Therefore [A] $\implies$ (A).

\item Assume [B] is the case. We have the following trace:
\begin{align*}
&\bsconf{G}{L}{\nameN}{\Sblock{S_1 .. S_n}}
\\
&\to
\bsconf{G}{L}{\nameN'}{\Sblock{S_1,\vec S_n}_L^\nameN}
\\
&\to^*
\bsconf{G_1}{L_1}{\nameN'}{\Sblock{\modeM_1,\vec S_n}_L^\nameN}
&\text{by CG and IH$\bsconf{G}{L}{\nameN'}{S_1}$ in [B]}
\\
&\to
\bsconf{G_1}{L_1\restrict L}{\nameN}{{\modeM_1}}
&\text{}
\end{align*}
Therefore [B] $\implies$ (B).
\end{enumerate}
Since [A] $\implies$ (A) and [B] $\implies$ (B), we have (1) $\implies$ (2).

\item \textbf{Variable Manipulation:} Let us consider rule \iref{tvardecl}. Let $S = \Svardecl{\vec x}{M}$. Assuming $\vec x \not\in\dom{L}$, for (1) we have:
\[
    \irule[VarDecl]{
      \bsconf{G}{L}{\nameN}{M}
      \bsevalexp
      \bsconf{G_1}{L_1}{\nameN}{\tuple{\vec v}}
    }{
      \bsconf{G}{L}{\nameN}{\Svardecl{\vec x}{M}}
      \bseval
      \bsconf{G_1}{L_1[\vec x \mapsto \vec v]}{\nameN}{\regk}
    }
\]
For (2) we have:
\begin{align*}
&\bsconf{G}{L}{\nameN}{\Svardecl{\vec x}{M}}
\\
&\to^*\bsconf{G_1}{L_1}{\nameN}{\Svardecl{\vec x}{\tuple{\vec v}}} &\text{by CG and IH$\bsconf{G}{L}{\nameN}{M}$}
\\
&\to\bsconf{G_1}{L_1[\vec x \mapsto \vec v]}{\nameN}{\regk}
\end{align*}
Therefore (1) $\implies$ (2) for \iref{tvardecl}. Cases \iref{vardecl}, \iref{assign} and \iref{tassign} are similar to above, but with singleton values $v$ for \iref{vardecl} and \iref{assign}, and assumption $\vec x \in\dom{L}$ for \iref{vardecl} and \iref{assign}.

\item \textbf{Conditional Branching:} Let us consider rule \iref{iftrue}. Let $S = \Scond{M}{{ S}}$. For (1) we have:
\[
    \irule[IfTrue]{
      \bsconf{G}{L}{\nameN}{M}
      \bsevalexp
      \bsconf{G_0}{L_0}{\nameN}{\true}
      \\\\
      \bsconf{G_0}{L_0}{\nameN}{{ S}}
      \bseval
      \bsconf{G_1}{L_1}{\nameN}{\modeM}
    }{
      \bsconf{G}{L}{\nameN}{\Scond{M}{{ S}}}
      \bseval
      \bsconf{G_1}{L_1}{\nameN}{\modeM}
    }
\]
For (2) we have:
\begin{align*}
&\bsconf{G}{L}{\nameN}{\Scond{M}{{ S}}}
\\
&\to^*\bsconf{G_0}{L_0}{\nameN}{\Scond{\true}{{ S}}}
&&\text{by CG and IH$\bsconf{G}{L}{\nameN}{M}$}
\\
&\to\bsconf{G_0}{L_0}{\nameN}{S}
\\
&\to^*\bsconf{G_1}{L_1}{\nameN}{\modeM}
&&\text{by CG and IH$\bsconf{G_0}{L_0}{\nameN}{{S}}$}
\end{align*}
Therefore (1) $\implies$ (2) for \iref{iftrue}. Case \iref{iffalse} is a simpler version of the above with $\false$ instead of $\true$, $\regk$ in place of $\modeM$, and no application of IH$\bsconf{G_0}{L_0}{\nameN}{{ S}}$.

\item \textbf{Switch Branching:} Cases for \iref{switchcase} and \iref{switchdef} are similar to the case for \iref{iftrue}, but with $v$ in place of $\true$, and either $S_i$ or $S_d$ in place of $S$ if $v = c_i$ or $v \not\in\{c_1, .. , c_n\}$ respectively.

\item \textbf{Loop Initialisation:} Let $\Sfor{\Sblock{S_1 .. S_n}}{M}{{S_{p}}}{{ S_{b}}}$. For (1) we have:
\[
    \irule[ForInit]{
      \bsconf{G_1}{L_1}{\nameN}{\Sblock{S_1 .. S_n\,\Sfor{\Sblock{}}{M}{S_{p}}{ S_{b}}}}
      \bseval
      \bsconf{G_2}{L_2}{\nameN}{\modeM}
      \\
      n > 0
    }{
      \bsconf{G_1}{L_1}{\nameN}{\Sfor{\Sblock{S_1 .. S_n}}{M}{{S_{p}}}{{ S_{b}}}}
      \bseval
      \bsconf{G_2}{L_2}{\nameN}{\modeM}
    }
\]
For (2) we have:
\begin{align*}
&\bsconf{G_1}{L_1}{\nameN}{\Sfor{\Sblock{S_1 .. S_n}}{M}{{S_{p}}}{{ S_{b}}}}
\\
&\to\bsconf{G_1}{L_1}{\nameN}{\Sblock{S_1 .. S_n\,\Sfor{\Sblock{}}{M}{S_{p}}{ S_{b}}}}
\\
&\to^*\bsconf{G_2}{L_2}{\nameN}{\modeM} \qquad\text{by IH$\bsconf{G_1}{L_1}{\nameN}{\Sblock{S_1 .. S_n\,\Sfor{\Sblock{}}{M}{S_{p}}{ S_{b}}}}$}
\end{align*}
Therefore (1) $\implies$ (2).

\item \textbf{Loop False:} Let $S = \Sfor{\Sblock{}}{M}{{ S_{p}}}{{ S_{b}}}$. For (1) we have:
\[
    \irule[ForFalse]{
      \bsconf{G}{L}{\nameN}{M}
      \bsevalexp
      \bsconf{G_1}{L_1}{\nameN}{\false}
    }{
      \bsconf{G}{L}{\nameN}{\Sfor{\Sblock{}}{M}{{ S_{p}}}{{ S_{b}}}}
      \bseval
      \bsconf{G_1}{L_1}{\nameN}{\regk}
    }
\]
For (2) we have:
\begin{align*}
&\bsconf{G}{L}{\nameN}{\Sfor{\Sblock{}}{M}{{ S_{p}}}{{ S_{b}}}}
\\
&\to
\bsconf{G}{L}{\nameN}{
\ssbreak{
    \Scond{M}{
      \Sblock{
        \sscont{S_b} S_p
        \Sfor{\Sblock{}}{M}{{S_{p}}}{{ S_{b}}}
      }
    }
  }
}
\\
&\to^*
\bsconf{G_1}{L_1}{\nameN}{
\ssbreak{
    \Scond{\false}{
      \Sblock{
        \sscont{S_b} S_p
        \Sfor{\Sblock{}}{M}{{S_{p}}}{{ S_{b}}}
      }
    }
  }
}\quad\text{by CG and IH$\langle M \mid \dots \rangle$}
\\
&\to
\bsconf{G_1}{L_1}{\nameN}{
\ssbreak{\regk}
}
\\
&\to
\bsconf{G_1}{L_1}{\nameN}{
\regk
}
\end{align*}

\item \textbf{Loop True:} Let $S = \Sfor{\Sblock{}}{M}{{ S_{p}}}{{ S_{b}}}$. For (1) we have the following cases:
\begin{enumerate}[{[}A{]}]
\item \textbf{Loop-Body Halt:}
\[
    \irule[ForHalt1]{
      (\modeM_1,\modeM_2) \in \{(\breakk,\regk) , (\leavek,\leavek)\}
      \\\\
      \bsconf{G}{L}{\nameN}{M}
      \bsevalexp
      \bsconf{G_1}{L_1}{\nameN}{\true}
      \\\\
      \bsconf{G_1}{L_1}{\nameN}{{ S_{b}}}
      \bseval
      \bsconf{G_2}{L_2}{\nameN}{\modeM_1}
    }{
      \bsconf{G}{L}{\nameN}{\Sfor{\Sblock{}}{M}{{ S_{p}}}{{ S_{b}}}}
      \bseval
      \bsconf{G_2}{L_2}{\nameN}{\modeM_2}
    }
\]
\item \textbf{Post-Iteration Halt:}
\[
    \irule[ForHalt2][forhalt2]{
      \bsconf{G}{L}{\nameN}{M}
      \bsevalexp
      \bsconf{G_1}{L_1}{\nameN}{\true}
      \\\\
      \bsconf{G_1}{L_1}{\nameN}{{ S_{b}}}
      \bseval
      \bsconf{G_2}{L_2}{\nameN}{\modeM}
      \\\modeM \not\in \{\leavek , \breakk\}
      \\\\
      \bsconf{G_2}{L_2}{\nameN}{{ S_{p}}}
      \bseval
      \bsconf{G_3}{L_3}{\nameN}{\leavek}
    }{
      \bsconf{G}{L}{\nameN}{\Sfor{\Sblock{}}{M}{{ S_{p}}}{{ S_{b}}}}
      \bseval
      \bsconf{G_3}{L_3}{\nameN}{\leavek}
    }
\]
\item \textbf{Loop Iteration:}
\[
    \irule[ForLoop][forloop]{
      \bsconf{G}{L}{\nameN}{M}
      \bsevalexp
      \bsconf{G_1}{L_1}{\nameN}{\true}
      \\\\
      \bsconf{G_1}{L_1}{\nameN}{{ S_{b}}}
      \bseval
      \bsconf{G_2}{L_2}{\nameN}{\modeM}
      \\\modeM \not\in \{\leavek , \breakk\}
      \\\\
      \bsconf{G_2}{L_2}{\nameN}{{ S_{p}}}
      \bseval
      \bsconf{G_3}{L_3}{\nameN}{\regk}
      \\\\
      \bsconf{G_3}{L_3}{\nameN}{\Sfor{\Sblock{}}{M}{{ S_{p}}}{{ S_{b}}}}
      \bseval
      \bsconf{G_4}{L_4}{\nameN}{\modeM_1}
    }{
      \bsconf{G}{L}{\nameN}{\Sfor{\Sblock{}}{M}{{ S_{p}}}{{ S_{b}}}}
      \bseval
      \bsconf{G_4}{L_4}{\nameN}{\modeM_1}
    }
\]
\end{enumerate}

For (2) we have:
\begin{align*}
&\bsconf{G}{L}{\nameN}{\Sfor{\Sblock{}}{M}{{ S_{p}}}{{ S_{b}}}}
\\
&\to
\bsconf{G}{L}{\nameN}{
\ssbreak{
    \Scond{M}{
      \Sblock{
        \sscont{S_b} S_p
        \Sfor{\Sblock{}}{M}{{S_{p}}}{{ S_{b}}}
      }
    }
  }
}
\\
&\to^*
\bsconf{G_1}{L}{\nameN}{
\ssbreak{
    \Scond{\true}{
      \Sblock{
        \sscont{S_b} S_p
        \Sfor{\Sblock{}}{M}{{S_{p}}}{{ S_{b}}}
      }
    }
  }
}&\text{by CG and IH$\bsconf{G}{L}{\nameN}{M}$}\\
&\to
\bsconf{G_1}{L}{\nameN}{
\ssbreak{
      \Sblock{
        \sscont{S_b} S_p
        \Sfor{\Sblock{}}{M}{{S_{p}}}{{ S_{b}}}
      }
  }
}
\end{align*}
From here we have the following cases:
\begin{enumerate}[(A)]
\item Assume [A] is the case, such that $(\modeM_1,\modeM_2) \in \{(\breakk,\regk) , (\leavek,\leavek)\}$. We continue the trace as follows:
\begin{align*}
&\to^*
\bsconf{G_2}{L_2}{\nameN}{
\ssbreak{
      \Sblock{
        \sscont{\modeM_1} S_p
        \Sfor{\Sblock{}}{M}{{S_{p}}}{{ S_{b}}}
      }
  }
}\quad\text{by CG, IH$\langle S_b \dots \rangle$ for [A]}
\\
&\to
\bsconf{G_2}{L_2}{\nameN}{
\ssbreak{
      \Sblock{
        \modeM_1 S_p
        \Sfor{\Sblock{}}{M}{{S_{p}}}{{ S_{b}}}
      }
  }
}
\\
&\to
\bsconf{G_2}{L_2}{\nameN}{
\ssbreak{\modeM_1}
}
\\
&\to
\bsconf{G_2}{L_2}{\nameN}{\modeM_2}
\end{align*}
Therefore [A] $\implies$ (A).

\item Assume [B] is the case, such that $\modeM \in \{\regk, \continuek\}$. We continue the trace as follows:
\begin{align*}
&\to^*
\bsconf{G_2}{L_2}{\nameN}{
\ssbreak{
      \Sblock{
        \sscont{\modeM} S_p
        \Sfor{\Sblock{}}{M}{{S_{p}}}{{ S_{b}}}
      }
  }
}
&&\text{by CG, IH$\langle S_b \dots \rangle$ for [B]}
\\
&\to
\bsconf{G_2}{L_2}{\nameN}{
\ssbreak{
      \Sblock{
        \regk S_p
        \Sfor{\Sblock{}}{M}{{S_{p}}}{{ S_{b}}}
      }
  }
}
\\
&\to
\bsconf{G_2}{L_2}{\nameN}{
\ssbreak{
      \Sblock{
        S_p
        \Sfor{\Sblock{}}{M}{{S_{p}}}{{ S_{b}}}
      }
  }
}
\\
&\to^*
\bsconf{G_3}{L_3}{\nameN}{
\ssbreak{
      \Sblock{
        \leavek~
        \Sfor{\Sblock{}}{M}{{S_{p}}}{{ S_{b}}}
      }
  }
}
&&\text{by CG, IH$\langle S_p \dots \rangle$ for [B]}
\\
&\to
\bsconf{G_3}{L_3}{\nameN}{
\ssbreak{\leavek}
}
\\
&\to
\bsconf{G_3}{L_3}{\nameN}{
\leavek
}
\end{align*}
Therefore [B] $\implies$ (B).

\item Assume [C] is the case, such that $\modeM \in \{\regk, \continuek\}$ and $\modeM \in \{\regk, \leavek\}$. We continue the trace as follows:
\begin{align*}
&\to^*
\bsconf{G_2}{L_2}{\nameN}{
\ssbreak{
      \Sblock{
        \sscont{\modeM} S_p
        \Sfor{\Sblock{}}{M}{{S_{p}}}{{ S_{b}}}
      }
  }
}\quad\text{by CG, IH$\langle S_b \dots \rangle$ for [C]}
\\
&\to
\bsconf{G_2}{L_2}{\nameN}{
\ssbreak{
      \Sblock{
        \regk S_p
        \Sfor{\Sblock{}}{M}{{S_{p}}}{{ S_{b}}}
      }
  }
}
\\
&\to
\bsconf{G_2}{L_2}{\nameN}{
\ssbreak{
      \Sblock{
        S_p
        \Sfor{\Sblock{}}{M}{{S_{p}}}{{ S_{b}}}
      }
  }
}
\\
&\to^*
\bsconf{G_3}{L_3}{\nameN}{
\ssbreak{
      \Sblock{
        \regk~
        \Sfor{\Sblock{}}{M}{{S_{p}}}{{ S_{b}}}
      }
  }
}\quad\text{by CG, IH$\langle S_p \dots \rangle$ for [C]}
\\
&\to
\bsconf{G_3}{L_3}{\nameN}{
\ssbreak{
      \Sblock{
        \Sfor{\Sblock{}}{M}{{S_{p}}}{{ S_{b}}}
      }
  }
}
\\
&\to^*
\bsconf{G_3}{L_3}{\nameN}{
\ssbreak{
      \Sblock{
        \modeM_1
      }
  }
}\quad^\text{by CG, IH$\bsconf{G_3}{L_3}{\nameN}{\Sfor{\Sblock{}}{M}{{ S_{p}}}{{ S_{b}}}}$ for [C]}
\\
&\to
\bsconf{G_3}{L_3}{\nameN}{
\ssbreak{
        \modeM_1
  }
}
\\
&\to
\bsconf{G_3}{L_3}{\nameN}{
\modeM_1
}\quad\text{since $\modeM_1 \neq \breakk$ by \cref{lem:break:cont} direction of on $\langle  \Sfor{\Sblock{}}{M}{{ S_{p}}}{{ S_{b}}} \dots \rangle$}
\end{align*}
Therefore [C] $\implies$ (C).
\end{enumerate}
Since [A,B,C] $\implies$ (A,B,C), we have (1) $\implies$ (2).
\item \textbf{Function Call:}
Let $S = \Efuncall{x}{M_1 , .. , M_n}$. For (1) we have:
\[
    \irule[FunCall]{
      \bsconf{G}{L}{\nameN}{M_n}
      \bsevalexp
      \bsconf{G_1}{L_1}{\nameN}{v_n}
      \\\\
      :
      \\\\
      \bsconf{G_{n-1}}{L_{n-1}}{\nameN}{M_1}
      \bsevalexp
      \bsconf{G_n}{L_n}{\nameN}{v_1}
      \\\\
      \nameN(x) = (\{y_1,..,y_n\},\{z_1,..,z_m\},S_{b})
      \\
      L' = \{\{y_1,..,y_n\} \mapsto \{v_1, .. , v_n\}\} \uplus \{\{z_1,..,z_m\} \mapsto 0\}
      \\\\
      \bsconf{G_n}{L'}{\nameN}{{S_b}}
      \bseval
      \bsconf{G''}{L''}{\nameN}{\modeM}
      \\\\
      S_{ret} =
      \nbox{
      \begin{cases}
      L''(z_1) &\text{if $m = 1$} \\
      \tuple{L''(z_1),..,L''(z_m)} &\text{if $m > 1$} \\
      \regk &\text{otherwise}
      \end{cases}
      }
    }{
      \bsconf{G}{L}{\nameN}{\Efuncall{x}{M_1 , .. , M_n}}
      \bsevalexp
      \bsconf{G''}{L_n}{\nameN}{S_{ret}}
    }
\]
For (2) we have:
\begin{align*}
&\bsconf{G}{L}{\nameN}{\Efuncall{x}{M_1 , .. , M_n}} \\
&\to^*\bsconf{G_1}{L_1}{\nameN}{\Efuncall{x}{M_1 , .. M_{n-1}, v_n}}
&&\quad\text{by CG and IH$\bsconf{G}{L}{\nameN}{M_n}$}\\
&\to^*\bsconf{G_2}{L_2}{\nameN}{\Efuncall{x}{M_1 , .. v_{n-1}, v_n}}
&&\quad\text{by CG and IH$\bsconf{G_1}{L_1}{\nameN}{M_{n-1}}$}\\
&\dots\\
&\to^*\bsconf{G_{n}}{L_{n}}{\nameN}{\Efuncall{x}{v_1 , .. , v_n}}
&&\quad\text{by CG and IH$\bsconf{G_{n-1}}{L_{n-1}}{\nameN}{M_1}$}\\
&\to\bsconf{G_{n}}{L'}{\nameN}{\ssframe{S_b}^{\vec z}_{L}}
&&\quad\text{where $\nameN(x) = (\vec y, \vec z, S_b)$}\\
&\to^*\bsconf{G''}{L''}{\nameN}{\ssframe{\modeM}^{\vec z}_{L_n}}
&&\quad\text{by CG and IH$\bsconf{G_{n}}{L'}{\nameN}{S_b}$}\\
&\to\bsconf{G''}{L_n}{\nameN}{S_{ret}}
&&
\begin{cases}
S_{ret} = L''(z_1)                     & m = 1\\
S_{ret} = \tuple{L''(z_1),..,L''(z_m)} & m > 1\\
S_{ret} = \regk                        & m = 0\\
\end{cases}
\end{align*}
Therefore (1) $\implies$ (2).
\item \textbf{Opcode Call:} \iref{opcall} is a simpler case to \iref{funcall} above.
\end{itemize}
\qed


\begin{lemma}[Semantic Entailment 2]\label{lem:2}
Given statement $S$ in source syntax, global environment $G$, state $L$, namespace $\nameN$, 
then for a configuration $\bsconf{G}{L}{\nameN}{S}$:
\begin{enumerate}[(1)]
\item $\bsconf{G}{L}{\nameN}{S} \to^* \bsconf{G'}{L'}{\nameN'}{S_{ret}}$
\item $\bsconf{G}{L}{\nameN}{S} \bseval \bsconf{G'}{L'}{\nameN'}{S_{ret}}$
\end{enumerate}
it is the case that (1) $\implies$ (2).
\end{lemma}
\subsection{Proof of Semantic Entailment 2}\label{appendix:entailment:2}
Consider a configuration $\bsconf{G}{L}{\nameN}{S}$ such that
\[\bsconf{G}{L}{\nameN}{S} \to^* \bsconf{G'}{L'}{\nameN'}{S_{r}}\]
Recall Lemma~\ref{lem:congruence} for congruence (CG).
By induction on the length of the reduction sequence and case analysis on the shape of the term, we have the following cases:
\begin{itemize}

\item \textbf{Loops:}
Let $S = \Sfor{S_i}{M}{ S_{p}}{ S_{b}}$. For (1) we have the following cases:
\begin{enumerate}[{[}A{]}]
\item \textbf{Empty initialisation block:}
Let $S_i = \Sblock{}$. We have:
\begin{flalign*}
&\bsconf{G}{L}{\nameN}{\Sfor{\Sblock{}}{M}{{ S_{p}}}{{ S_{b}}}} &\\
&\to\bsconf{G}{L}{\nameN}{
  \ssbreak{
    \Scond{M}{
      \Sblock{
        \sscont{S_b} S_p
        \Sfor{\Sblock{}}{M}{{S_{p}}}{{ S_{b}}}
      }
    }
  }
} &\text{}
\end{flalign*}
By case analysis on the evaluation of $M$:
\begin{enumerate}[{[}I{]}]
\item \textbf{False:}
Assume $\bsconf{G}{L}{\nameN}{M}\to^*\bsconf{G'}{L'}{\nameN}{\false}$. Then:
\begin{flalign*}
&\to^*\bsconf{G'}{L'}{\nameN}{
  \ssbreak{
    \Scond{\false}{
      \Sblock{
        \sscont{S_b} S_p
        \Sfor{\Sblock{}}{M}{{S_{p}}}{{ S_{b}}}
      }
    }
  }
} &\text{by CG}\\
&\to\bsconf{G'}{L'}{\nameN}{
  \ssbreak{\regk}
}\\
&\to\bsconf{G'}{L'}{\nameN}{\regk} &
\end{flalign*}

\item \textbf{True:}
Assume $\bsconf{G}{L}{\nameN}{M}\to^*\bsconf{G'}{L}{\nameN}{\true}$ (since expressions cannot modify the top-level $L$) and let $\nameN' = \funsof{\Sfor{\Sblock{}}{M}{{S_{p}}}{{ S_{b}}}}$. Then:
\begin{flalign*}
&\to^*\bsconf{G'}{L}{\nameN}{
  \ssbreak{
    \Scond{\true}{
      \Sblock{
        \sscont{S_b} S_p
        \Sfor{\Sblock{}}{M}{{S_{p}}}{{ S_{b}}}
      }
    }
  }
} &\text{by CG}\\
&\to\bsconf{G'}{L}{\nameN}{
  \ssbreak{
      \Sblock{
        \sscont{S_b} S_p
        \Sfor{\Sblock{}}{M}{{S_{p}}}{{ S_{b}}}
      }
  }
} &
\\
&\to\bsconf{G'}{L}{\nameN}{
  \ssbreak{
      \Sblock{
        \sscont{S_b} S_p
        \Sfor{\Sblock{}}{M}{{S_{p}}}{{ S_{b}}}
      }_{L}^\nameN
  }
} &\text{since $\nameN' = \varnothing$}
\end{flalign*}
By case analysis on the evaluation of $S_b$:
\begin{enumerate}[{[}a{]}]
\item \textbf{Continue:}
Assume $\bsconf{G'}{L}{\nameN}{S_b}\to^*\bsconf{G''}{L''}{\nameN}{\continuek}$:
\begin{flalign*}
&\to^*\bsconf{G''}{L''}{\nameN}{
  \ssbreak{
      \Sblock{
        \sscont{\continuek} S_p
        \Sfor{\Sblock{}}{M}{{S_{p}}}{{ S_{b}}}
      }_{L}^\nameN
  }
} &\text{by CG}
\\
&\to\bsconf{G''}{L''}{\nameN}{
  \ssbreak{
      \Sblock{
        \regk S_p
        \Sfor{\Sblock{}}{M}{{S_{p}}}{{ S_{b}}}
      }_{L}^\nameN
  }
} &
\\
&\to^*\bsconf{G''}{L''}{\nameN}{
  \ssbreak{
      \Sblock{
        S_p
        \Sfor{\Sblock{}}{M}{{S_{p}}}{{ S_{b}}}
      }_{L}^\nameN
  }
} &
\end{flalign*}
By cases on evaluation of $S_p$, given Definition~\ref{def:halt:restrict} and \cref{lem:break:cont} with \cref{lem:1} (i.e. can only evaluate to \leavek and \regk):
\begin{enumerate}[{[}i{]}]
\item \textbf{Leave:}
Assume $\bsconf{G''}{L''}{\nameN}{S_p}\to^*\bsconf{G'''}{L'''}{\nameN}{\leavek}$:
\begin{flalign*}
&\to^*\bsconf{G'''}{L'''}{\nameN}{
  \ssbreak{
      \Sblock{
        \leavek\,
        \Sfor{\Sblock{}}{M}{{S_{p}}}{{ S_{b}}}
      }_{L}^\nameN
  }
} &\text{by CG}
\\
&\to^*\bsconf{G'''}{L'''\restrict L}{\nameN}{
  \ssbreak{\leavek}
} &
\\
&\to\bsconf{G'''}{L'''}{\nameN}{
  \leavek
} &\text{by Lem.\ref{lem:l:invariance:loops}}
\end{flalign*}

\item \textbf{Regular:}
Assume $\bsconf{G''}{L''}{\nameN}{S_p}\to^*\bsconf{G_3}{L_3}{\nameN}{\regk}$ and $\bsconf{G_3}{L_3}{\nameN}{\Sfor{\Sblock{}}{M}{{S_{p}}}{{ S_{b}}}}
\to^*
\bsconf{G_4}{L_4}{\nameN}{\modeM}$:
\begin{flalign*}
&\to^*\bsconf{G_3}{L_3}{\nameN}{
  \ssbreak{
      \Sblock{
        \regk\,
        \Sfor{\Sblock{}}{M}{{S_{p}}}{{ S_{b}}}
      }_{L}^\nameN
  }
} &
\\
&\to\bsconf{G_3}{L_3}{\nameN}{
  \ssbreak{
      \Sblock{
        \Sfor{\Sblock{}}{M}{{S_{p}}}{{ S_{b}}}
      }_{L}^\nameN
  }
} &\\
&\to^*\bsconf{G_4}{L_4}{\nameN}{
  \ssbreak{
      \Sblock{
        \modeM
      }_{L}^\nameN
  }
}&\\
&\to^*\bsconf{G_4}{L_4}{\nameN}{\modeM} &\text{by Lem.\ref{lem:l:invariance:loops} \& \ref{lem:bound:l}}
\end{flalign*}
where $\modeM \in \{\regk,\leavek\}$ by Def.\ref{def:halt:restrict}, \cref{lem:break:cont} with \cref{lem:1}.
\end{enumerate}

\item \textbf{Regular:}
Assume $\bsconf{G'}{L'}{\nameN}{S_b}\to^*\bsconf{G''}{L''}{\nameN}{\regk}$:
\begin{flalign*}
&\to^*\bsconf{G'}{L'}{\nameN}{
  \ssbreak{
      \Sblock{
        \sscont{\regk} S_p
        \Sfor{\Sblock{}}{M}{{S_{p}}}{{ S_{b}}}
      }
  }
} &\\
&\to\bsconf{G'}{L'}{\nameN}{
  \ssbreak{
      \Sblock{
        \regk S_p
        \Sfor{\Sblock{}}{M}{{S_{p}}}{{ S_{b}}}
      }
  }
} &
\end{flalign*}
Then proceeds similar to case [A][II][a].

\item \textbf{Regular:}
Assume $\bsconf{G'}{L'}{\nameN}{S_b}\to^*\bsconf{G''}{L''}{\nameN}{\breakk}$:
\begin{flalign*}
&\to^*\bsconf{G''}{L''}{\nameN}{
  \ssbreak{
      \Sblock{
        \sscont{\breakk} S_p
        \Sfor{\Sblock{}}{M}{{S_{p}}}{{ S_{b}}}
      }
  }
} &\\
&\to\bsconf{G''}{L''}{\nameN}{
  \ssbreak{
      \Sblock{
        \breakk\,
        \Sfor{\Sblock{}}{M}{{S_{p}}}{{ S_{b}}}
      }
  }
} &\\
&\to^*\bsconf{G''}{L''}{\nameN}{
  \ssbreak{ \breakk }
} & \\
&\to\bsconf{G''}{L''}{\nameN}{
  \regk
} &
\end{flalign*}

\end{enumerate}
\end{enumerate}

\item \textbf{Non-empty initialisation block:}
Let $S_i = \Sblock{S_1..S_n}$ and let $\nameN = \funsof{S_1..S_n\Sfor{\Sblock{}}{M}{{ S_{p}}}{{ S_{b}}}}$. We have:
\begin{flalign*}
&\bsconf{G}{L}{\nameN}{\Sfor{\Sblock{S_1..S_n}}{M}{{ S_{p}}}{{ S_{b}}}} &
\\
&\to\bsconf{G}{L}{\nameN}{\Sblock{S_1..S_n\Sfor{\Sblock{}}{M}{{ S_{p}}}{{ S_{b}}}}} &
\\
&\to\bsconf{G}{L}
{\nameN}
{\Sblock
{S_1..S_n\Sfor{\Sblock{}}{M}{{ S_{p}}}{{ S_{b}}}}_L^\nameN
} &\text{since $\nameN' = \varnothing$}
\end{flalign*}
From here, by case analysis on the evaluation of $\Sblock{S_1..S_n}$:
\begin{enumerate}[{[}I{]}]
\item \textbf{Leave:}
Assume $\bsconf{G}{L}{\nameN}{S_i}\to^*\bsconf{G'}{L'}{\nameN}{\leavek}$. Then:
\begin{flalign*}
&\to^*\bsconf{G'}{L'}{\nameN}
{
\Sblock{\leavek\,\Sfor{\Sblock{}}{M}{{ S_{p}}}{{ S_{b}}}}_L^\nameN
} &\text{by CG}
\\
&\to\bsconf{G'}{L'\restrict L}{\nameN}{\leavek}
\end{flalign*}

\item \textbf{Regular:}
Assume $\bsconf{G}{L}{\nameN}{S_i}\to^*\bsconf{G'}{L'}{\nameN}{\regk}$. Then:
\begin{flalign*}
&\to^*\bsconf{G'}{L'}{\nameN}{
\Sblock{\regk\,\Sfor{\Sblock{}}{M}{{ S_{p}}}{{ S_{b}}}}_L^\nameN
} &\text{by CG}
\\
&\to\bsconf{G'}{L'}{\nameN}{
\Sblock{\Sfor{\Sblock{}}{M}{{ S_{p}}}{{ S_{b}}}}_L^\nameN
} &\text{}
\end{flalign*}
which by Lemma~\ref{lem:l:invariance:loops} proceeds similar to case [A] above.
\end{enumerate}
\end{enumerate}
For (2) we have the following cases:
\begin{itemize}
\item[(A)(I)] \textbf{Assume [A][I].} We have:
\[
    \irule[ForFalse]{
        \irule[\normalfont{by IH on [A][I]}]{
        }{
          \bsconf{G}{L}{\nameN}{M}
          \bsevalexp
          \bsconf{G'}{L}{\nameN}{\false}
        }
    }{
      \bsconf{G}{L}{\nameN}{\Sfor{\Sblock{}}{M}{{ S_{p}}}{{ S_{b}}}}
      \bseval
      \bsconf{G'}{L}{\nameN}{\regk}
    }
\]
Therefore [A][I] $\implies$ (A)(I).

\item[(A)(II)] \textbf{Assume [A][II].} By IH on [A][II], we know:
\[
 \bsconf{G}{L}{\nameN}{M}
 \bsevalexp
 \bsconf{G'}{L}{\nameN}{\true} \qquad\text{{(A)(II)}}
\]
Then we have the following cases:
\begin{enumerate}
\item[(a)] \textbf{Assume [a].} By IH on [A][II][a], we know:
\begin{align*}
\bsconf{G'}{L}{\nameN}{S_b}\bseval\bsconf{G''}{L''}{\nameN}{\continuek}
&\qquad\text{{(A)(II)(a)}}
\end{align*}

\item[~\quad(i)] \textbf{Assume [a][i].} By IH on [A][II][a][i], we know:
\begin{align*}
\bsconf{G''}{L''}{\nameN}{S_p}\bseval\bsconf{G'''}{L'''}{\nameN}{\leavek}
&\qquad\text{{(A)(II)(a)(i)}}
\end{align*}
We have:
\[
    \irule[ForHalt2]{
      \text{(A)(II)}
      \\
      \text{(A)(II)(a)}
      \\
      \text{(A)(II)(a)(i)}
    }{
      \bsconf{G}{L}{\nameN}{\Sfor{\Sblock{}}{M}{{ S_{p}}}{{ S_{b}}}}
      \bseval
      \bsconf{G'''}{L'''}{\nameN}{\leavek}
    }
\]
Therefore [A][II][a][i] $\implies$ (A)(II)(a)(i).
\item[~\quad(ii)] \textbf{Assume [a][ii].} By IH on [A][II][a][ii], we know:
\begin{align*}
\bsconf{G''}{L''}{\nameN}{S_p}\bseval\bsconf{G_3}{L_3}{\nameN}{\regk}
&\qquad\text{{(A)(II)(a)(ii)1}}
\\
\bsconf{G_3}{L_3}{\nameN}{\Sfor{\Sblock{}}{M}{{S_{p}}}{{ S_{b}}}}
\bseval
\bsconf{G_4}{L_4}{\nameN}{\modeM}
&\qquad\text{{(A)(II)(a)(ii)2}}
\end{align*}
We have:
\[
    \irule[ForLoop]{
      \text{(A)(II)}
      \\
      \text{(A)(II)(a)}
      \\\\
      \text{(A)(II)(a)(ii)1}
      \\
      \text{(A)(II)(a)(ii)2}
    }{
      \bsconf{G}{L}{\nameN}{\Sfor{\Sblock{}}{M}{{ S_{p}}}{{ S_{b}}}}
      \bseval
      \bsconf{G_4}{L_4}{\nameN}{\modeM}
    }
\]
Therefore [A][II][a][ii] $\implies$ (A)(II)(a)(ii).
\item[(b)] \textbf{Assume [b].} Proven like cases for (A)(II)(a) above.

\item[(c)] \textbf{Assume [c].} By IH on [A][II], we know:
\[
 \bsconf{G'}{L}{\nameN}{{ S_{b}}}
 \bseval
 \bsconf{G''}{L''}{\nameN}{\breakk} \qquad\text{{(A)(II)(c)}}
\]
We have:
\[
    \irule[ForHalt1]{
      \text{(A)(II)}
      \\
      \text{(A)(II)(c)}
    }{
      \bsconf{G}{L}{\nameN}{\Sfor{\Sblock{}}{M}{{ S_{p}}}{{ S_{b}}}}
      \bseval
      \bsconf{G''}{L''}{\nameN}{\regk}
    }
\]
Therefore [A][II][c] $\implies$ (A)(II)(c).
\end{enumerate}

\item[(B)(I)] \textbf{Assume [B][I].} We have:
\[
\irule[ForInit]{
  \irule[Block]{
    \irule[]{
    \irule[\normalfont{by IH on [B][I]}]{
    }{
        \bsconf{G}{L}{\nameN}{S_1, .. ,S_n}
        \bseval
        \bsconf{G'}{L'}{\nameN}{\leavek}
    }
    }{
        \bsconf{G}{L}{\nameN}{S_1 , .. , S_n,\Sfor{\Sblock{}}{M}{S_{p}}{ S_{b}}}
        \bseval
        \bsconf{G'}{L'}{\nameN}{\leavek}
    }
  }{
    \bsconf{G}{L}{\nameN}{\Sblock{S_1 .. S_n\,\Sfor{\Sblock{}}{M}{S_{p}}{ S_{b}}}}
    \bseval
    \bsconf{G'}{L' \restrict L}{\nameN}{\leavek}
  }
}{
  \bsconf{G}{L}{\nameN}{\Sfor{\Sblock{S_1 .. S_n}}{M}{{S_{p}}}{{ S_{b}}}}
  \bseval
  \bsconf{G'}{L' \restrict L}{\nameN}{\leavek}
}
\]
Therefore [B][I] $\implies$ (B)(I).
\item[(B)(II)] \textbf{Assume [B][II].} Proven like cases for (A) above.
\end{itemize}

All cases [A],[B] $\implies$ (A),(B), therefore (1) $\implies$ (2).

\item \textbf{Blocks and branching:}
Proven by the same mechanisms already seen in the case for loops above. Loops involve both blocks and conditional statements and thus the proof of the former covers a substantial portion of the latter's. The correspondence is also more clear by inspection.

\item \textbf{Variables:}
The correspondence from small-step to big-step semantics is clear by inspection of the rules. The proof proceeds by simple application of IH on the configuration containing $M$ for a given assignment $x := M$.

\item \textbf{Expressions:}
Since the case for dereferencing identifiers is obvious by inspection, and the case for opcode calls is a simplification of the case for function calls, let us focus on function calls. For a term $S = \Efuncall{x}{M_1, .., M_n}$, assuming
\[
\bsconf{G_{n-i}}{L_{n-i}}{\nameN}{M_i}
\to^*
\bsconf{G_{n-i+1}}{L_{n-i+1}}{\nameN}{v_i} \qquad\text{[A]}
\]
for each $i \in [1,n]$, for (1) we have the trace:
\begin{align*}
&\bsconf{G_0}{L_0}{\nameN}{\Efuncall{x}{M_1 , .. , M_n}}\\
&
\to^*
\bsconf{G_1}{L_1}{\nameN}{\Efuncall{x}{M_1 , .. , v_n}}\\
&
\to^*
\bsconf{G_{n}}{L_{n}}{\nameN}{\Efuncall{x}{v_1 , .. , v_n}}\\
&
\to^*
\bsconf{G_{n}}
{\{ \vec y \mapsto \vec v \} \uplus \{ \vec z \mapsto 0 \}}
{\nameN}
{\ssframe{S_b}^{\vec z}_{L_n}}&\nameN(x) = (\vec y, \vec z, S_b)
\end{align*}
assuming
\[
\bsconf{G_{n}}
{\{ \vec y \mapsto \vec v \} \uplus \{ \vec z \mapsto 0 \}}
{\nameN}
{S_b}
\to^*
\bsconf{G''}
{L''}
{\nameN}
{\modeM}
\qquad\text{[B]}
\]
we continue trace as follows:
\begin{align*}
&
\to^*
\bsconf
{G''}
{L''}
{\nameN}
{\ssframe{\modeM}^{\vec z}_{L_n}}\\
&
\to^*
\bsconf
{G''}
{L_n}
{\nameN}
{\tuple{L''(z_1) , .. , L''(z_m))}}&\text{$\vec z = z_1, .., z_m$, $m > 1$}
\end{align*}
For (2), we first have by IH on [A]:
\[
\bsconf{G_{n-i}}{L_{n-i}}{\nameN}{M_i}
\bseval
\bsconf{G_{n-i+1}}{L_{n-i+1}}{\nameN}{v_i} \qquad\text{(A)}
\]
Following, by IH on [B]:
\[
\bsconf{G_{n}}
{\{ \vec y \mapsto \vec v \} \uplus \{ \vec z \mapsto 0 \}}
{\nameN}
{S_b}
\bseval
\bsconf{G''}{L''}{\nameN}{\modeM}
\qquad\text{(B)}
\]
We therefore have the following:
\[
    \irule[FunCall]{
      \bsconf{G_0}{L_0}{\nameN}{M_n}
      \bsevalexp
      \bsconf{G_1}{L_1}{\nameN}{v_n}
      \\\\
      :
      \\\\
      \bsconf{G_{n-1}}{L_{n-1}}{\nameN}{M_1}
      \bsevalexp
      \bsconf{G_n}{L_n}{\nameN}{v_1}
      \\\\
      \bsconf{G_n}{\{ \vec y \mapsto \vec v \} \uplus \{ \vec z \mapsto 0 \}}{\nameN}{{S_b}}
      \bseval
      \bsconf{G''}{L''}{\nameN}{\modeM}
    }{
      \bsconf{G_0}{L_0}{\nameN}{\Efuncall{x}{M_1 , .. , M_n}}
      \bsevalexp
      \bsconf{G''}{L_n}{\nameN}{\tuple{L''(z_1),..,L''(z_m)}}
    }
\]
Therefore, the case for $n,m > 1$ holds.
Cases for $n=1$, $m=1$ and $m=0$ are simplifications of the above. Note for opcode calls that cases for $\mathcal{M}$ are not necessary to inspect because we are only considering correctly terminating reductions.
\end{itemize}
With all cases done, we have shown (1) $\implies$ (2).
\qed


  \section{Redundant Contexts}
  In this section we define and prove properties of our small-step semantics that will become useful to remove redundant evaluation contexts.

\begin{lemma}[Redundant frames]\label{lem:redundant}
Given a statement $S$, local stores $L_1 .. L_n$, $E \in \{ \ssbreak{ \cdot} , \Sblock{\cdot}_{L_i}^{\nameN''}\}$ for $i > 1$,
and $\dom{L_1} \subseteq \dom{L_2} \subseteq .. \subseteq \dom{L_n}$ then
\begin{align*}
&\bsconf{G}{L}{\nameN}{E_1[{\Sblock{S}_{L_1}^{\nameN''}}]}
\to^*
\bsconf{G'}{L'}{\nameN'}{\modeM}
\\
&\iff
\bsconf{G}{L}{\nameN}{E_m[..E_1[{\Sblock{S}_{L_1}^{\nameN''}}]..]}
\to^*
\bsconf{G'}{L'}{\nameN'}{\modeM}
\end{align*}
where {$E_m[..E_1[\cdot]..]$} is a context composed only of $E$'s nested in descending order such that either (1) $\ssbreak{ \cdot} \not\in \{E_1 , .. , E_m \}$ or (2) $E_1 = \ssbreak{ \cdot}$ .
\end{lemma}
\begin{proof}
This follows by inspection and structural induction on the reductions.
For case (1), it is clear that $L'' \restrict L_1 = ((L'' \restrict L_1) \restrict L_2) .. \restrict L_n$, with everything else in the configurations being the same in the reduction sequence (for one nesting and then for arbitrarily many). For case (2), if
$\bsconf{G}{L}{\nameN}{S}
\to^*
\bsconf{G'}{L'}{\nameN'}{\breakk}$,
then this is caught by $E_1$ and converted into \regk. Afterwards, $\ssbreak{\cdot}$ become nullable (as they do not have any effect), which reduces case (2) to case (1), which also applies in every other case.
\qed
\end{proof}

\begin{corollary}[Block dropping]\label{cor:block:drop}
From Lemma~\ref{lem:redundant} we can drop sequences of block contexts that appear nested in descending order (i.e. innermost $L$ subsumes all outer ones).
\end{corollary}

\begin{corollary}[Break dropping]\label{cor:break:drop}
From Lemma~\ref{lem:redundant} we can see that any interleaving of $\ssbreak{ \cdot}$ and $\Sblock{\cdot}_{L}^{\nameN}$ will not affect the result so long as one $\ssbreak{ \cdot}$ is present. This means we can drop break contexts that occur nested with block contexts.
\end{corollary}


  \section{Yul Objects}
  \label{sec:obj}
  To use Yul with the solidity compiler, one must use the Yul object notation. Since this notation is external to Yul statements and an official semantics has not been provided, we shall leave out discussing their semantics formally. Instead, this section is an overview. The grammar for objects is as follows.
  \[\begin{array}{r@{}r@{\,}c@{\,}l}
    \textsc{\ObjSet: } & O & \mis & \yulobj{c_{\Str}}{\Sblock{S^*}}{D^*} \\
    \textsc{\DataSet: } & D & \mis & O \mor \yuldat{c_{\Str}}{c} \\
  \end{array}\]
That is, an object is in essence defined by a triple $(c_{\Str}, S, D^*)$ where $c_{\Str}$ is a string literal that \emph{names} the object, a \emph{code fragment} $S$, and a list of \emph{data fragments} that allows for nesting of objects.
In addition to the object notation, the EVM dialect for the Solidity compiler extends the set of opcodes with data functions that enable access to named object data elements nested in $D^*$.
  \[\begin{array}{r@{}r@{\,}c@{\,}l}
    \textsc{\DataFuns:}   & \delta & \mis & \datacopyk \mor \dataoffsetk \mor \datasizek\\
    \textsc{\Exp:} & M & \mis & \dots \mor \Efuncall{\datasizek}{c_\Str}
                                      \mor \Efuncall{\dataoffsetk}{c_\Str}
                                      \mor \Efuncall{\datacopyk}{M,M,M}\\
  \end{array}\]
Data functions $\delta$ are variants of opcodes \textsf{sizeof}, \textsf{offset}, and \textsf{codecopy}, and behave analogously --- \datasizek and \dataoffsetk must be provided a string that identifies an element in the object, and \datacopyk behaves just like \textsf{codecopy} in the EVM dialect. Addressing object data is done with a dot notation similar to that of object-oriented programming languages such as Java.
Lastly, the EVM dialect with objects introduces a function \textsf{memoryguard}, which assists the compiler and the optimiser by indicating that memory access is restricted to the range defined by interval $[0, \textsf{size})$ for $\textsf{memoryguard(size)}$.
To illustrate the use of $\delta$ functions, the following is a typical pattern in Yul objects produced by compiling Solidity smart contracts.
\begin{lstlisting}
/// @use-src 0:"HelloWorld.sol"
object "HelloWorld_20" {
  code {
    ...
    let _1 := allocate_unbounded()
    codecopy(_1, dataoffset("HelloWorld_20_deployed"),
                 `datasize("HelloWorld_20_deployed"))
    ...
  }
  /// @use-src 0:"HelloWorld.sol"
  object "HelloWorld_20_deployed" {
    code { ... }
    data ".metadata" hex"a26469706..."
  }
}
\end{lstlisting}
Comments contain line annotations added by the compiler. The top-level object, named \textsf{"HelloWorld\_20"}, contains a sub-object, named \textsf{"HelloWorld\_20\_deployed"}, which the top-level code fragment copies into a newly allocated portion of memory specified in \textsf{\_1} after calling a Yul function \textsf{allocate\_unbounded()}. Besides code, \textsf{"HelloWorld\_20\_deployed"} contains a special data fragment \textsf{".metadata"} containing a hex literal which cannot be accessed from the code fragment. Names for special data fragments always start with a dot \cite{yul}.

{
Although a semantics is not officially provided, one can estimate a semantics for objects that incorporates our semantics for statements based on the behaviour of the Solidity compiler. Consider the following big-step rule.
  \[\begin{array}{@{}cllll@{}}
    \irule[Object][object]{
      \bsconf{G}{\varnothing}{\varnothing}{S}
      \bseval_\Delta
      \bsconf{G_1}{L_1}{\nameN_1}{\modeM}
    }{
      \objconf{\yulobj{c_{\Str}}{S}{D^*}}{G}
      \bseval_\ObjSet
      \objconf{\modeM}{G_1}
    }
  \end{array}\]
where $\Delta$ is an environment that encodes the list of data fragments $D^*$ in some manner parametric to the dialect. Intuitively, this says that the evaluation of an object is that of its code fragment within the context of its data fragment. Within $S$, only $\delta$ functions would require accessing $\Delta$. Note that the specific interpretation of $\Delta$ depends on the dialect. More specifically, depending on how program code is stored, $\Delta$ may additionally need to map the corresponding address offset for each $D \in \Delta$ such that \textsf{codecopy} can access the named data via \dataoffsetk and \datasizek. Lastly, we shall leave the definition of the dot naming notation to the chosen instance of $\Delta$.
}



\end{document}
